\documentclass[12pt]{article}
\usepackage[dvips]{epsfig}
\usepackage{natbib,bm}
\usepackage{graphics,epsfig,pstricks}
\usepackage{graphicx}
\usepackage{xcolor}
\usepackage{tikz}
\usepackage{longtable}
\usepackage{subfigure}
\usepackage{amssymb}
\usepackage{multirow}
\usetikzlibrary{positioning, arrows.meta}
\usepackage{placeins}
\usepackage{epstopdf}
\usepackage{chngcntr}
\usepackage{apptools}
\usepackage{diagbox}
\usepackage{amsmath,amssymb}
\usepackage{multirow}
\usepackage{float}
\usepackage{listings}
\usepackage{caption}
\usepackage{authblk}
\usepackage{comment}
\usepackage{algorithm}
\usepackage{algpseudocode}
\usepackage{lipsum}
\usepackage{setspace}
\usepackage{tikz,xcolor,hyperref}
\floatstyle{plaintop}
\restylefloat{table}
\AtAppendix{\counterwithin{lemma}{section}}
\AtAppendix{\counterwithin{theorem}{subsection}}
\newtheorem{theorem}{Theorem}
\newtheorem{lemma}{Lemma}
\usepackage[top=1in,right=1in,bottom=1in,left=1in,includefoot,includehead]{geometry}
\UseRawInputEncoding
\newcommand{\be}{\begin{eqnarray}}
	\newcommand{\ee}{\end{eqnarray}}
\newcommand{\ba}{\begin{eqnarray*}}
	\newcommand{\ea}{\end{eqnarray*}}

\definecolor{lime}{HTML}{A6CE39}
\DeclareRobustCommand{\orcidicon}{%
	\begin{tikzpicture}
		\draw[lime, fill=lime] (0,0) 
		circle [radius=0.18] 
		node[white] {{\fontfamily{qag}\selectfont \tiny ID}};
		\draw[white, fill=white] (-0.0625,0.095) 
		circle [radius=0.007];
	\end{tikzpicture}
	\hspace{-2mm}
}

\foreach \x in {A, ..., Z}{%
	\expandafter\xdef\csname orcid\x\endcsname{\noexpand\href{https://orcid.org/\csname orcidauthor\x\endcsname}{\noexpand\orcidicon}}
}

\newtheorem{remark0}{Remark}
\newtheorem{fact0}{Fact}
\newtheorem{example0}{Example}
\newtheorem{corollary0}{Corollary}
\newtheorem{proposition0}{Proposition}
\newtheorem{conjecture0}{Conjecture}

\newenvironment{remark}{\begin{remark0} \mbox{}\rm}{\end{remark0}}

\newenvironment{proof}{\noindent{\bf Proof.~} \mbox{} \rm}{ \hfill \fbox{}}

\newcommand{\expect}{\mbox{\rm I\kern-.20em E}}
\newcommand{\reals}{\mbox{\rm I\kern-.20em R}}
\newcommand{\sreals}{\mbox{\small \rm I\kern-.20em R}}

\newcommand{\twocases}[6]

\newif\ifshowchanges
 \showchangesfalse  % 最终版：不显示标注
\ifshowchanges
  \usepackage[draft]{changes}
  \definechangesauthor[name={Your Name}, color=blue]{HP}
\else
  \makeatletter
  \providecommand{\added}[2][]{#2}
  \providecommand{\deleted}[2][]{}
  \providecommand{\replaced}[3][]{#2} % 若用到了 \replaced
  \providecommand{\definechangesauthor}[2][]{}
  \makeatother
\fi

\title{A Bayesian Optimal Phase II Design for Randomized Immunotherapy Trials with Delayed Treatment Effects}
\author[]{Zhongheng Cai}
\author[]{Haitao Pan
	\thanks{Corresponding author: Haitao.Pan@stjude.org}}
\affil[]{Department of Biostatistics, St. Jude Children's Research Hospital, Memphis, USA}
\date{}
\begin{document}	
	\maketitle
	\begin{abstract}
Immunotherapy has transformed cancer treatment, yet its delayed therapeutic effects often lead to non-proportional hazards, rendering many conventional phase II designs underpowered and prone to type I error inflation. To address this issue, we propose a novel Bayesian Optimal Phase II design (DTE-BOP2) that explicitly models the uncertainty in the separation timing of treatment effect. The treatment separation timepoint (denoted by $S$) is endowed with a truncated-Gamma prior, whose parameters can be elicited from experts or inferred from historical data, with default settings available when prior knowledge is scarce. Built upon the BOP2 framework \citep{Zhou:2017,Zhou:2020}, our design retains operational simplicity while incorporating type I error control and maintaining the power. Extensive simulations demonstrate that DTE-BOP2 uniformly controls type I error at the nominal level across a wide range of treatment effect separation timepoint $S$. We further observe that the power decreases monotonically as $S$ increases. Importantly, we find that the power is primarily driven by the relative magnitude of treatment benefit before and after the separation time, i.e., the ratio of medians, rather than their absolute values. Compared to the original BOP2, the piecewise weighted log-rank, and the conventional log-rank tests, DTE-BOP2 achieves higher power with smaller sample sizes while preserving type I error robustness across plausible delay scenarios. An open-source \textsf{R} package, \texttt{DTEBOP2} (CRAN), with detailed vignettes, enables investigators to implement the design and analyse phase-II trials exhibiting delayed treatment effects.
\end{abstract}

  \added[id=HP]{Bayesian adaptive design; Bayesian decision theory; phase-II clinical trials; immunotherapy; delayed treatment effect; non-proportional hazards; truncated-Gamma prior}

\section{Introduction}

Delayed treatment effects (DTEs) present major challenges in the design and analysis of randomized clinical trials, particularly in immuno-oncology. In studies evaluating agents such as immune checkpoint inhibitors or therapeutic cancer vaccines, clinical benefits often do not appear immediately but emerge only after a latency period. For example, in the phase III CheckMate 017 trial \citep{Julie:2015}, which assessed nivolumab in patients with squamous-cell non-small-cell lung cancer (NSCLC), the progression-free survival (PFS) curves for the treatment and control groups overlapped for approximately 2.5 months before beginning to diverge. A similar delayed separation was observed in another phase III NSCLC trial, complicating interim analyses and futility decisions \citep{PazAres:2019}. Likewise, in a phase III study of ipilimumab for advanced melanoma, an overall survival (OS) benefit did not become apparent until about four months after randomization \citep{Hodi:2010}.

These delays are biologically plausible, as immunotherapies typically require time to stimulate the host immune system before producing measurable clinical effects. As a result, early survival curves for the treatment and control arms often overlap before eventual separation. This temporal shift in treatment effect can undermine the performance of conventional trial designs and analysis methods--such as those relying on Kaplan-Meier estimates and log-rank tests--by reducing statistical power and increasing the risk of false-negative conclusions. Collectively, these examples underscore the clinical relevance and growing prevalence of DTEs in oncology, emphasizing the need for alternative design and analysis strategies that appropriately account for delayed effects.

To accommodate delayed treatment effects (DTEs), a number of frequentist methods have been proposed to modify trial design and analysis, particularly in time-to-event settings. For instance, \citet{Zhang:2009} introduced a lag-time model where the treatment effect is assumed to initiate after a prespecified time threshold. Depending on whether treatment discontinuation is treated as censoring (Non-ITT) or all patients are followed regardless of adherence (ITT), they derived distinct sample size formulas, improving upon standard approaches such as Schoenfeld’s method. \citet{Sit:2016} later extended this framework by rescaling time using the control arm’s cumulative hazard, enabling more accurate and computationally efficient power calculations under fixed delays.

Subsequent work further refined hypothesis testing strategies to improve power under delayed effects. \citet{Xu:2017} proposed a piecewise weighted log-rank test that upweights later events—more likely to reflect true treatment benefit—and introduced sample size formulas based on both analytical (APPLE) and simulation-based (SEPPLE) methods. While these methods can be effective when the delay duration is well-characterized, their performance declines when prior knowledge is uncertain or miscalibrated. \citet{Wu:2019} addressed some of these concerns by revising sample size formulas under a fixed alternative hypothesis, improving accuracy in cases with small effect sizes or unbalanced arms.

To further address uncertainty in delay duration, robust testing strategies have been developed. The Maximin Efficiency Robust Test (MERT) proposed by \citet{Ye:2019} optimizes power across a range of plausible delays, while \citet{Xu:2018} introduced the GPW-log-rank test to account for subject-specific lag variability, offering asymptotically optimal weighting. These approaches have also been extended to group sequential designs, where standard interim monitoring can be biased under delayed effects if proportional hazards are assumed \citep{Zhang:2016}.

Despite these advances, frequentist methods often rely on strong assumptions--such as fixed or known delay durations, proportional hazards, or homogeneous treatment effects--that may not hold in early-phase immunotherapy trials where uncertainty is high. \added[id=HP]{This is especially problematic in early-phase trials, where limited patient numbers, short accrual windows, and evolving treatment paradigms make rigid assumptions about effect timing particularly risky.} Moreover, many of these approaches offer limited flexibility for adaptive decision-making during the course of a trial. These limitations underscore the potential value of Bayesian methods, which naturally incorporate uncertainty in key parameters and support a probabilistic framework for trial planning and analysis.

From a Bayesian perspective, relatively few methodological developments have addressed phase II trial designs that explicitly accommodate delayed treatment effects (DTEs). A notable exception is the work of \citet{Sals:2024}, who propose a Bayesian framework specifically tailored for clinical trial planning in settings where DTEs are expected, such as immuno-oncology. Their approach centers on the use of \textit{assurance}, a Bayesian counterpart to conventional power calculations. Unlike traditional power, which assumes fixed parameter values (e.g., a specific treatment effect size), assurance incorporates uncertainty by averaging the power over a prior distribution, thus quantifying the probability of trial success under the current state of knowledge. This makes assurance--also known as average power or predictive power--particularly well-suited for realistic trial evaluations. In their model, survival times in both treatment and control arms follow Weibull distributions, with the hazard functions assumed identical up to a delay time $T$, after which the treatment arm begins to show a potential benefit. They further develop a simulation-based algorithm to compute assurance under such a delayed effect scenario.

While this framework provides valuable insight for trial planning, it does not offer a fully operational design for conducting trials with interim monitoring and decision-making under DTE. One widely adopted design in phase II settings is the Bayesian Optimal Phase II (BOP2) design, proposed by \citet{Zhou:2020}, which is appreciated for its simplicity, efficiency, and ability to control type I error while maximizing power. However, the original BOP2 formulation assumes the proportional hazards—assumption that may not hold in trials involving immunotherapies or other agents with delayed onset of action. Applying standard BOP2 to such settings may therefore lead to underpowered or miscalibrated decisions, especially at interim looks.

To address this methodological gap, we propose the \textbf{Delayed Treatment Effect BOP2} (\textbf{DTE-BOP2}) design, which extends the BOP2 framework by explicitly modeling the unknown delay time using a truncated gamma prior. This prior is incorporated into the Bayesian decision-making process, allowing the design to dynamically adapt to uncertainty in the separation timing of treatment effects. DTE-BOP2 preserves rigorous control of type I and II error rates, improves power when delays are present, and can reduce the required sample size--making it a practical and robust tool for phase II immunotherapy trials. To facilitate implementation, we provide an open-source R package, \texttt{DTEBOP2}, which includes all necessary functions for design calibration, operating characteristic evaluation, and decision rule generation.

The remainder of this paper is organized as follows. \added[id=HP]{Section 2 formalizes the DTE-BOP2 model and the decision framework, including an algorithm for two-stage sample size calculation. Section 3 demonstrates the application of our method using a motivating trial example and illustrates its implementation using the \texttt{DTEBOP2} package. Section 4 presents a comprehensive simulation study comparing DTE-BOP2 with existing methods under varying delay scenarios and prior choices. Section 5 concludes with a discussion of practical implications and future extensions.} \deleted[id=HP]{Section 2 introduces the model and decision rule of the proposed design, along with a sample size determination algorithm for a two-stage setting. Section 3 presents a real trial example to illustrate the application of the method and implementation using the \texttt{DTEBOP2} package and compare DTE-BOP2 with the frequentist approach. Section 4 compares the operating characteristics of DTE-BOP2 with the original BOP2 design and investigate the sensitivity of the choice of priors. Section 5 is a discussion of our method.}

 \section{Method}
\subsection{Prior Specification for the Delay Time $S$}

Suppose that in an interim analysis, $2n$ patients have been enrolled and an equal allocation ratio between the treatment arm and the control arm. For the $i$-th patient, let $T_{ij}$ denote the time to event, where $j = 0$ corresponds to the control arm and $j = 1$ corresponds to the treatment arm. The observed time, $Z_{ij}$, is defined as $Z_{ij} = \min(T_{ij}, C_{ij})$, where $C_{ij}$ represents the administrative censoring time. The event indicator is given by $\delta_{ij} = I(T_{ij} \leq C_{ij})$ for $i = 1, \dots, n$, and $d_{j} = \sum_{i=1}^n \delta_{ij}$ is the number of events in the control and treatment arms, respectively.

For patients in the control arm, we assume $T_{i0} \sim \text{Exp}(\mu_0)$ with mean $\mu_0$ and probability density function $f_{T_{i0}}(t|\mu_0)=\mu_0^{-1}\exp(\added[id=HP]{-}t/\mu_0)$. In the treatment arm, the event time $T_{i1}$ follows a piecewise exponential distribution with a delayed treatment effect. Specifically, given a separation timepoint $S$, the probability density function (p.d.f.) of $T_{i1}$ is:

\begin{equation}
    f_{T_{i1}}(t) = 
    \begin{cases} 
    \frac{1}{\mu_0} \exp(- \frac{t}{\mu_0}), & 0 < t < S, \\
    \frac{1}{\mu_1} \exp(-\frac{t}{\mu_1} - (\frac{1}{\mu_0} -\frac{1}{\mu_1})S), & t \geq S,
    \end{cases}
    \label{eq2.1-1}
\end{equation}

Here, $\mu_0^{-1}$ represents the hazard rate for the interval $(0, S)$, while $\mu_1^{-1}$ represents the hazard rate for the interval $[S, +\infty)$.

To \deleted[id=HP]{model}\added[id=HP]{construct} the prior distribution of $S$, we assume that it is bounded between $L$ and $U$. A truncated Gamma prior is employed for $S$:
\begin{equation}
    \pi(S) \propto \text{Gamma}(s_1, s_2) I(L, U),
    \label{eq2.1_2}
\end{equation}
where $I(L, U)$ indicates that $S$ is restricted to the interval $[L, U]$. The values for $L$ and $U$, and the shape and rate parameters $s_1$ and $s_2$ in \eqref{eq2.1_2}, can be specified \added[id=HP]{in a default manner or guided} by expert opinion. \added[id=HP]{If elicitation is preferred, we may construct the prior distribution of $S$ by posing the following questions to domain experts:}

\begin{enumerate}
\item[1.] \textit{"Based on your knowledge, what are reasonable lower and upper bounds for the delay time $S$ in the population survival curves?"}
\item[2.] \textit{"Based on your clinical experience or intuition, are you able to estimate one or two typical characteristics of $S$ (e.g., its median, expected value, or a likely range where 95\% of values might fall)?"}
\end{enumerate}

Typically, the first question\deleted[id=HP]{--eliciting plausible lower and upper bounds for the delay time $S$--} is relatively straightforward for domain experts to answer. \added[id=HP]{Eliciting plausible lower and upper bounds for the delay time $S$ is often relatively straightforward for domain experts.} \added[id=HP]{For the second question, even rough or qualitative approximations can be valuable. For example, an expert might suggest that the benefit of treatment tends to appear around the third or fourth cycle, or that it is unlikely to occur before two months or after six months. Such impressions can be interpreted into approximate statistical summaries (e.g., mean, median, or quantiles) to guide prior specification.} When such information is available, we fit the truncated Gamma prior by minimizing the discrepancy between the elicited and theoretical values of the summary statistics.

\deleted[id=HP]{Specifically, we estimate the parameters $s_1$ and $s_2$ by minimizing the following weighted least squares (WLS) criterion:}
\added[id=HP]{To calibrate the truncated Gamma prior for $S$, we estimate the parameters $s_1$ and $s_2$ by minimizing a weighted least squares (WLS) criterion that measures the discrepancy between the theoretical summaries of the prior and the elicited quantities from domain experts. This approach accommodates multiple types of summary statistics and allows each to be weighted according to its perceived reliability or importance.}

\begin{align*}
& w_1\sum_{i=1}^n(\text{mean}(S)-\text{mean from Expert } i)^2 + w_2\sum_{i=1}^n(\text{median}(S)-\text{median from Expert } i)^2 + \\
& w_3\sum_{i=1}^n(\text{sd}(S)-\text{sd from Expert } i)^2 + w_4\sum_{i=1}^n(\text{quantile}(S,0.025)-\text{0.025 quantile from Expert } i)^2  +\\
&w_5\sum_{i=1}^n(\text{quantile}(S,0.975)-\text{0.975 quantile from Expert } i)^2,
\end{align*}

where \( S \sim \text{Gamma}(s_1, s_2)I(L, U) \), and \( w_1,\dots,w_5 \) denote user-specified weights reflecting the relative importance of each elicited quantity.

\deleted[id=HP]{Importantly, not all experts need to provide every statistic. The procedure remains valid even when only a subset of elicited quantities is available, allowing for flexible and practical prior specification. After fitting, the estimated prior can be graphically compared against the elicited values to solicit feedback or refinement from experts.}

\added[id=HP]{This procedure is flexible in that not all experts need to provide every summary statistic; the WLS criterion can be computed using whichever quantities are available. Once the prior is fitted, its implied distribution can be visualized and compared against the elicited values to facilitate feedback and potential refinement.}

\subsection{Likelihood and Posterior Inference}
\added[id=HP]{Once the prior for the delay time $S$ is specified, the next step is to derive the likelihood function and the posterior distributions for the key model parameters, conditional on $S$. This involves handling a piecewise exponential model in the treatment arm due to the delayed effect. To simplify the likelihood formulation, we follow \cite{Han:2014} and define the total time-on-test (TTOT) over any interval $(t_1, t_2]$ as the total observed follow-up time for patients within this interval in the treatment arm:}

\begin{equation}
    \text{TTOT}(t_1, t_2) = \sum_{i=1}^{n} \max(0, \min(Z_{i1}, t_2) - t_1).
    \label{eq2.1_3}
\end{equation}

\paragraph{Piecewise Event Counts and Interval Partitioning.}
\added[id=HP]{Let $d_{01}$ and $d_{11}$ denote the number of events observed in the intervals $(0, S]$ and $(S, \tau]$, respectively, where $\tau = \max_i Z_{i1}$. We denote $I_0 = 0$, $I_1 = S$, and $I_2 = \tau$ for notational convenience.}

\paragraph{Joint Likelihood Function.}
\added[id=HP]{Assuming exponential distributions with rates $\mu_0^{-1}$ and $\mu_1^{-1}$, the likelihood function based on all observed data is:}

\begin{equation}
    \begin{aligned}
    &f(z_{10}, \cdots, z_{n0}, z_{11}, \cdots, z_{n1} | \mu_0, \mu_1) = \\
    &\quad \mu_0^{-d_0} \exp \left( -\frac{\sum_{i=1}^{n} z_{i0}}{\mu_0} \right) \cdot
    \mu_0^{-d_{01}} \exp \left( - \frac{\text{TTOT}(I_0, I_1)}{\mu_0} \right) \cdot
    \mu_1^{-d_{11}} \exp \left( - \frac{\text{TTOT}(I_1, I_2)}{\mu_1} \right),
    \end{aligned}
    \label{eq2.1_4}
\end{equation}

\paragraph{Prior Specification and Posterior Updates.}
\added[id=HP]{We assume independent conjugate inverse-Gamma priors for $\mu_0$ and $\mu_1$:}
\begin{equation}
    \pi(\mu_0) \sim \text{Inv-Gamma}(a_0, b_0), \quad \pi(\mu_1) \sim \text{Inv-Gamma}(a_1, b_1),
    \label{eq2.1_5}
\end{equation}

\added[id=HP]{Then, conditional on $S$, the posterior distributions for $\mu_0$ and $\mu_1$ are:}
\begin{align}
    \pi(\mu_0 | \mathcal{D}_n, S) &\sim \text{Inv-Gamma}(a_0 + d_0 + d_{01}, b_0 + \text{TTOT}(I_0, I_1) + \sum_{i=1}^{n} z_{i0}), \nonumber \\
    \pi(\mu_1 | \mathcal{D}_n, S) &\sim \text{Inv-Gamma}(a_1 + d_{11}, b_1 + \text{TTOT}(I_1, I_2)).
    \label{eq2.1_6}
\end{align}

\paragraph{Posterior Distribution for Median Survival Time.}

\added[id=HP]{Although the model is parameterized using the mean survival time $\mu_j$ ($j=0,1$), it is standard practice in oncology and clinical literature to report survival outcomes using the median survival time, denoted by $\tilde{\mu}_j$. This is largely due to convention: most clinical studies present median survival estimates, as they are easily interpretable from Kaplan-Meier curves and commonly used in regulatory and medical reporting.}

\added[id=HP]{Under the exponential distribution, the median and mean survival times are related by $\tilde{\mu}_j = \mu_j \log 2$. This relationship, combined with the scaling property of the inverse-gamma distribution--namely, if $X \sim \text{Inv-Gamma}(a, b)$, then $cX \sim \text{Inv-Gamma}(a, cb)$—allows us to derive posterior distributions for $\tilde{\mu}_0$ and $\tilde{\mu}_1$ directly.}

\added[id=HP]{If we assign the prior $\tilde{\mu}_j \sim \text{Inv-Gamma}(a_j, b_j \log 2)$ for $j=0,1$, then the posterior distributions are given by:}

\begin{align}
    \pi(\tilde{\mu}_0 | \mathcal{D}_n, S) &\sim \text{Inv-Gamma}\left(a_0 + d_0 + d_{01}, \log 2 \cdot \left(b_0 + \text{TTOT}(I_0, I_1) + \sum_{i=1}^{n} z_{i0} \right) \right), \nonumber \\
    \pi(\tilde{\mu}_1 | \mathcal{D}_n, S) &\sim \text{Inv-Gamma}\left(a_1 + d_{11}, \log 2 \cdot \left(b_1 + \text{TTOT}(I_1, I_2) \right) \right),
    \label{eq:median_post}
\end{align}
\added[id=HP]{and the two posterior distributions are conditionally independent given $S$.}

\added[id=HP]{The above conditional independence arises from the model structure. Although the treatment arm data contribute to both $\tilde{\mu}_0$ and $\tilde{\mu}_1$, the contributions are disjoint once $S$ is fixed: the early part $(0, S]$ informs $\tilde{\mu}_0$, and the later part $(S, \tau]$ informs $\tilde{\mu}_1$, where $\tau = \max_i Z_{i1}$. The control arm data are entirely used to inform $\tilde{\mu}_0$. Since the priors for $\tilde{\mu}_0$ and $\tilde{\mu}_1$ are specified independently and the data components influencing each parameter are non-overlapping, the resulting posterior distributions are conditionally independent given $S$.}

\added[id=HP]{This structure simplifies posterior computation: samples of $\tilde{\mu}_0$ and $\tilde{\mu}_1$ can be drawn independently, and quantities such as $P(\tilde{\mu}_1 > \tilde{\mu}_0)$ can be efficiently estimated from independent posterior draws.}

\paragraph{\added[id=HP]{Connecting Clinical and Model-Specific Survival Times.}}
\added[id=HP]{To operationalize the proposed Bayesian framework using clinically meaningful parameters, it is necessary to establish a mapping between model-based median survival $\tilde{\mu}_j$ and the overall clinical median survival $\bar{\mu}_j$. We now introduce this conversion relationship.}

\added[id=HP]{Specifically, when delayed treatment effect is assumed as in Equation~\eqref{eq2.1-1}, the relationship between $\bar{\mu}_1$ and $\tilde{\mu}_1$ becomes:}
\begin{equation}
    \bar{\mu}_1 = 
    \begin{cases}
        \tilde{\mu}_0 & \text{if } \tilde{\mu}_0 < S, \\
        \left(1 - \frac{\tilde{\mu}_1}{\tilde{\mu}_0}\right) S + \tilde{\mu}_1 & \text{if } \tilde{\mu}_0 \geq S,
    \end{cases}
    \label{relation}
\end{equation}
\added[id=HP]{where $S$ is the prespecified separation timepoint of the delayed treatment effect.}

\added[id=HP]{For example, suppose the standard-of-care has a median survival of $\bar{\mu}_0 = 4$ months, and the expected improvement under treatment is $\bar{\mu}_1 = 7$ months, with $S = 2$ months. Solving the equation:}
\[
\left(1 - \frac{\tilde{\mu}_1}{4}\right) \cdot 2 + \tilde{\mu}_1 = 7,
\]
\added[id=HP]{yields the implied $\tilde{\mu}_1 = 10$ months. This conversion ensures that the Bayesian model operates on consistent parameters, while aligning with clinically meaningful quantities.}

\subsection{Bayesian Group Sequential Design}

\deleted[id=HP]{Consider a trial that includes $R$ interim analyses with one final analysis and it will be conducted when the number of enrolled patients reaches $2n_1 < \cdots < 2n_R<2n_{R+1}=2N$, where $2N$ is maximum number of patients. The decision to continue or terminate the trial for futility at each interim analysis is based on a comparison of the median survival time $\tilde{\mu}_0$  and $\tilde{\mu}_1$. Specifically, suppose that $2n_r$, $r=1,...,R+1$ patients have enrolled in two arms, given separation timepoint $S$ and dataset $\mathcal{D}_{n_r}$, the trial is terminated for futility in the interim analysis, if the following $A(S,n_r)=1$:}

\added[id=HP]{We consider a Bayesian group sequential design that includes $R$ interim analyses and one final analysis, conducted at increasing total sample sizes $2n_1 < \cdots < 2n_R < 2n_{R+1} = 2N$, where $2N$ is the maximum total number of patients. At each interim look $r = 1, \dots, R$, the decision to continue or terminate the trial is based on the posterior comparison between the median survival times $\tilde{\mu}_0$ and $\tilde{\mu}_1$.}

\added[id=HP]{Specifically, given a fixed separation timepoint $S$ and data $\mathcal{D}_{n_r}$ from $2n_r$ enrolled patients (with equal allocation), we define a futility stopping rule as:}

\begin{equation}
   \deleted[id=HP]{A(S,n_r)=1(P(\tilde{\mu}_1 < \tilde{\mu}_0 \mid \mathcal{D}_{n_r}, S)> C(n_r)),}
   \added[id=HP]{A(S,n_r) = \mathbb{I} \left( P(\tilde{\mu}_1 < \tilde{\mu}_0 \mid \mathcal{D}_{n_r}, S) > C(n_r) \right),}
    \label{eq2.2_1}
\end{equation}

\deleted[id=HP]{here $C(n_r)$ represents the posterior probability threshold, which depends on the interim sample size $n_r$. Although early stopping for efficacy may also be considered, in this study, we focus on stopping for futility. For the final analysis, if $A(S,N)=0$, we reject the null hypothesis.}

\added[id=HP]{where $\mathbb{I}(\cdot)$ is the indicator function, and $C(n_r)$ is a pre-specified threshold function that depends on the interim sample size $n_r$. That is, the trial stops for futility at interim analysis $r$ if the posterior probability that the treatment median survival is worse than the control exceeds the threshold $C(n_r)$.}

\added[id=HP]{Although early stopping for efficacy could also be considered, we focus here on futility stopping only. At the final analysis ($r = R + 1$), if the trial has not been stopped early (i.e., $A(S, n_i) = 0,i=1,\cdots, R+1$), we reject the null hypothesis.}

The calculation of $P(\tilde{\mu}_1 < \tilde{\mu}_0 \mid \mathcal{D}_{n_r}, S)$, based on the posterior distribution in equation \eqref{eq2.1_5}, involves a double integral, which may pose computational challenges. To simplify this calculation, we introduce the following lemma:

\begin{lemma}
    \begin{equation}
            P(\tilde{\mu}_1< \tilde{\mu}_0 \mid \mathcal{D}_{n_r}, S) = P\left(K < \frac{b_0 + \text{TTOT}(I_0, I_1) + \sum_{i=1}^{n_r} z_{i0}}{b_0 + b_{1} + \sum_{i=1}^{n_r} Z_{i0} + \text{TTOT}(I_0, I_2)}\right),
            \label{formula}
    \end{equation}
    where $K \sim \text{Beta}(a_0 + d_0 + d_{01}, a_1 + d_{11})$.
\end{lemma}

\begin{proof}
    \begin{equation}
    \begin{split}
        P(\tilde{\mu}_1 <\tilde{\mu}_0 \mid \mathcal{D}_{n_r}, S)&=P(\mu_1^{-1} > \mu_0^{-1} \mid \mathcal{D}_{n_r}, S)\\
        &= P\left( \text{Gamma}(a_1 + d_{11}, b_1 + \text{TTOT}(I_1, I_2)) \right. \\
        &\quad\quad> \text{Gamma}\left(a_0 + d_0 + d_{01}, b_0 + \text{TTOT}(I_0, I_1) + \sum_{i=1}^{n_r} z_{i0}) \right) \\
        &= P\left( \frac{\text{Gamma}(a_0 + d_0 + d_{01}, 1)}{\text{Gamma}(a_0 + d_0 + d_{01}, 1) + \text{Gamma}(a_1 + d_{11}, 1)} \right. \\
        &\quad\quad\left.< \frac{b_0 + \text{TTOT}(I_0, I_1) + \sum_{i=1}^{n_r} z_{i0}}{b_0 + b_1 + \sum_{i=1}^{n_r} z_{i0} + \text{TTOT}(I_0, I_2)} \right).
    \end{split}
    \end{equation}

    By the properties of the Gamma distribution, we know that
    \[
    \frac{\text{Gamma}(a_0 + d_0 + d_{01}, 1)}{\text{Gamma}(a_0 + d_0 + d_{01}, 1) + \text{Gamma}(a_1 + d_{11}, 1)} \sim \text{Beta}(a_0 + d_0 + d_{01}, a_1 + d_{11}).
    \]
    This completes the proof.
\end{proof}

Lemma~1 implies that under the null hypothesis (\( H_0 \)), the proposed decision rule is explicitly dependent on the value of \( S \), and as a result, the type I error rate becomes a function of \( S \). This stands in contrast to the traditional BOP2 design, in which the type I error rate is fixed and does not vary with \( S \). This dependence necessitates additional care in controlling the type I error across the plausible range of \( S \).

\paragraph{Controlling Average Type I Error and Power under Uncertain $S$.}

\deleted[id=HP]{Given that the separation timepoint \( S \) is assumed to be random within a specified range, it is natural to focus on controlling the \textit{average} type I error and \textit{average} power, rather than these metrics at a fixed timepoint \( S \). This approach inherently addresses the uncertainty associated with the separation timepoint, providing robust operating characteristics suitable for practical applications. The average type I error and power are defined as}
\added[id=HP]{Given that the separation timepoint \( S \) is uncertain and modeled as a random variable within a pre-specified range, it is more appropriate to control the \textit{average} type I error and \textit{average} power, rather than these metrics at a fixed $S$. This marginal approach naturally accounts for the variability in $S$, offering robust operating characteristics across a plausible range of clinical scenarios.}

\added[id=HP]{Specifically, let \( A(S, n_r) = 1 \) denote early termination for futility at interim sample size \( n_r \). Then, the average type I error and average power are defined as follows:}
\begin{align}
    \text{Average Type I error}&=\int \prod_{r=1}^{R+1}(1-A(S,n_r))\pi(S)dS,\text{when $H_0$ is true}\label{averagetypei}\\
    \text{Average power}&=\int \prod_{r=1}^{R+1}(1-A(S,n_r))\pi(S)dS,\text{when $H_1$ is true}\label{averagepower}
\end{align}

\added[id=HP]{This formulation ensures that the decision criteria are averaged over the possible realizations of $S$, rather than being conditional on a single value. In fact, the concept of average error rates has been closely linked to the notion of “assurance” in prior literature (e.g., \cite{Sals:2024}), and has been advocated as a practically meaningful criterion when dealing with design uncertainty.}

\paragraph{\added[id=HP]{Selection of Tuning Parameters via Simulation}}

\added[id=HP]{Having defined the average type I error and power metrics, we now describe how the tuning parameters $(\lambda, \gamma)$ are selected to meet these design criteria.}

\deleted[id=HP]{As described by \cite{Zhou:2017}, the stopping threshold \( C(n_r) = 1 - \lambda \left(\frac{n_r}{N}\right)^\gamma \) in the BOP2 design is adaptive, becoming more stringent as the trial progresses. Early in the trial, a more lenient efficacy requirement is used to prevent premature termination due to limited data, allowing time for additional data collection. However, as the trial nears completion, stricter stopping criteria are applied to ensure timely discontinuation of treatments that show no benefit. The tuning parameters \( \lambda \) and \( \gamma \) are calibrated to maximize statistical power while controlling the type I error rate at a pre-specified level.}

\added[id=HP]{Following the BOP2 framework \cite{Zhou:2017}, the adaptive threshold function \( C(n_r) = 1 - \lambda \left(\frac{n_r}{N}\right)^\gamma \) becomes progressively more stringent as the trial advances. This design prevents premature termination due to sparse early data by using a lenient threshold initially, and later enforces stricter futility stopping as more information accumulates. The parameters \( \lambda \) and \( \gamma \) define the shape of this function and are calibrated to control the \textit{average} type I error at a pre-specified level while maximizing \textit{average} power.}

\added[id=HP]{We determine the optimal values of $(\lambda, \gamma)$ through a simulation-based algorithm as detailed in Algorithm~\ref{A11}.}

\begin{algorithm}
\caption{\added[id=HP]{Optimal Tuning Parameters Search via Simulation}}
\label{A11}
\begin{algorithmic}[1]
\State \textbf{Input:} \added[id=HP]{Target type I error rate $\alpha$, hypothesized median survival times $\bar{\mu}_0$ and $\bar{\mu}_1$ under $H_0$ and $H_1$}
\State \textbf{Output:} \added[id=HP]{Optimal $(\lambda, \gamma)$ pair that satisfies type I error control and maximizes power}
\State \textbf{Step 1:} \added[id=HP]{Specify $H_0$, $H_1$, and target type I error level $\alpha$}
\For{each candidate pair $(\lambda, \gamma)$ on a predefined grid}
    \State Simulate \added[id=HP]{10,000} datasets under both $H_0$ and $H_1$
    \State Under $H_0$: simulate \added[id=HP]{$T_{i0} \sim \text{Exp}(\mu_0)$}
    \State Under $H_1$: simulate \added[id=HP]{$T_{i1}$ using Equation~\eqref{eq2.1-1}, with $S$ sampled from the prior in Equation~\eqref{eq2.1_2}}
    \State Apply stopping rule in Equation~\eqref{eq2.2_1}
    \State Compute \added[id=HP]{average type I error and average power}
\EndFor
\State \textbf{Step 2:} \added[id=HP]{Select $(\lambda, \gamma)$ pairs with type I error $\leq \alpha$}
\State \textbf{Step 3:} \added[id=HP]{Among those, choose the pair that yields the highest power}
\end{algorithmic}
\end{algorithm}

\begin{remark}
    While the tuning parameters \((\lambda, \gamma)\) are optimized to control the \emph{average} type I error across the delay interval \( S \in [L, U] \), a common concern arises from the fact that the true value of \( S \) is unknown in practice. Consequently, there is a possibility that the type I error may slightly exceed the nominal level \(\alpha\) at certain values of \( S \), particularly near the boundaries.

To mitigate this, our design explicitly imposes type I error control not only in expectation but also at the endpoints \( S = L \) and \( S = U \). This dual-calibration strategy enhances robustness by ensuring both local (boundary) and global (average) control, as formalized in Equation~\eqref{eq2.2_1}.

As demonstrated in the type I error trajectory shown in Figure~\ref{fig:app_1} (Appendix~A), the error rate exhibits only modest variation across the plausible range of \( S \). This validates the practical effectiveness of enforcing control at the boundaries as a surrogate for uniform control across the entire interval. Although this constraint may lead to a slight reduction in statistical power, our simulation results indicate that the loss is minimal and remains well within acceptable operational margins.

\end{remark}

\paragraph{\added[id=HP]{Specifying the Most Likely Delay Time.}}
\added[id=HP]{In addition to calibrating decision rules based on the distribution of $S$, it is often useful to specify a representative or most likely value of the delay time in practice. This quantity is denoted as \( S_{\text{likely}} \), and can be used to define \( \tilde{\mu}_1 \) corresponding to a target clinical median survival via Equation~\eqref{relation}.}

\added[id=HP]{This value can be elicited in several ways. One option is to define \( S_{\text{likely}} \) as the mean of the truncated Gamma prior for \( S \), or as a simple midpoint of the elicited interval \([L, U]\), i.e., \( S_{\text{likely}} = (L + U)/2 \). Alternatively, when historical survival data are available, \( S_{\text{likely}} \) may be estimated empirically as the time point at which the Kaplan--Meier curves for the treatment and control arms begin to separate. This estimate provides a clinically grounded basis for calibrating the design under realistic delay scenarios.}

\subsection{Sample Size Determination in Two-Stage Design}

\added[id=HP]{While the original BOP2 framework provides practical decision rules, they do not include procedures for determining the necessary sample size to achieve desired operating characteristics,} even in the two-stage setting. In this subsection, we introduce a novel method for determining both the interim and final sample sizes when the design consists of two stages, filling an important gap in the current literature. \added[id=HP]{We focus on the two-stage setting in this section due to its balance between computational tractability and practical utility. Extending the proposed sample size methodology to multiple interim looks would require high-dimensional integration over nested stopping boundaries, which is computationally intensive and beyond the scope of this paper.}

The sample size for a two-stage design is determined by minimizing the following weighted expected sample size function, while ensuring that the power is at least $1-\beta$ when \( S \in [L,U] \). \added[id=HP]{This design aims to balance the need to minimize patient exposure to futile treatments and to maximize access to effective therapies.} \deleted[id=HP]{To balance two competing objectives, minimizing patient exposure to futile treatments and maximizing access to effective treatments, we define the function as follows:}

\begin{equation}
 EN = w \frac{PS_{H_0}\times n_1 + (1 - PS_{H_0})(n - n_1)}{n} +
(1 - w) \left(1 - \frac{PS_{H_1} \times n_1 + (1 - PS_{H_1})(n - n_1)}{n} \right)
\label{samplesize}
\end{equation}

Here, $PS_{H_i},i=0,1$ is the probability of stopping the trial under $H_i$, $n_1$ is the interim sample size per arm, $n$ is the final sample size per arm and  $w\in[0,1]$ is a tuning parameter representing the trade-off between two goals:
\begin{itemize}
    \item When $w=1$: Equation \eqref{samplesize} simplifies to 
    \[
    EN=E(N|H_0)=\frac{PS_{H_0}\times n_1 + (1 - PS_{H_0})(n - n_1)}{n}
    \]
    This expression reflects the expected sample size under the null hypothesis. Minimizing this value emphasizes early stopping when the treatment is futile, thereby reducing unnecessary patient exposure to ineffective therapies.
    
    \item When $w=0$: Equation \eqref{samplesize} becomes 
    \[
    EN=1-E(N|H_1)=1 - \frac{PS_{H_1} \times n_1 + (1 - PS_{H_1})(n - n_1)}{n}
    \]
    \replaced[id=HP]{This reflects an emphasis on maximizing the number of patients who can access a truly effective treatment. By minimizing $1-E(N|H_1)$, the design encourages continuation when the treatment is promising, avoiding premature termination.}{In this case, the focus shifts to maximizing the number of patients who can access a truly effective treatment. By minimizing $1-E(N|H_1)$, we effectively encourage continuation of the trial when the treatment is promising, avoiding premature termination that could deny access to beneficial therapy.}
    
    \item \added[id=HP]{By varying $w$, the design allows flexible prioritization between safety (minimizing exposure to ineffective treatment) and efficacy (maximizing access to effective treatment), while ensuring adequate statistical power when the true signal is weak.}
\end{itemize}

\added[id=HP]{The sample size computation algorithm is summarized below in Algorithm \ref{Al2}.}

\begin{algorithm}[H]
\caption{\deleted[id=HP]{Calculating Interim and Final Sample Size with Optimal Design Parameters} \added[id=HP]{Sample Size Determination for Two-Stage DTE-BOP2 Design}}
\label{Al2} 
\begin{algorithmic}[1]
\State \textbf{Input:} Desired type I error rate $\alpha$, power $1 - \beta$, overall median survival time $\bar{\mu}_0$, $\bar{\mu}_1$, follow-up duration, \deleted[id=HP]{accural} \added[id=HP]{accrual} rate.
\State \textbf{Output:} Optimized interim and maximum sample sizes, optimal design parameters.

\State \textbf{Step 1: Estimate Initial Maximum Sample Size} \\
Compute:
\[
\text{nmax} = \frac{4(z_{1-\alpha/2} + z_{1-\beta})^2}{P(\log(\tilde{\mu}_1/\deleted[id=HP]{\tilde{\mu}_2}\added[id=HP]{\tilde{\mu}_0}))^2} \cdot \frac{\tilde{\mu}_0 + \tilde{\mu}_1}{\tilde{\mu}_0 + \bar{\mu}_1}
\]
to obtain a preliminary estimate of the maximum sample size, where $P$ is the overall event rate. The first part is sample size using Schoenfeld formula~\citep{Scho:1981}. The second\deleted[id=HP]{one} \added[id=HP]{part} is the adjustment factor \added[id=HP]{accounting for} \deleted[id=HP]{with considering}the delayed treatment effect.

\State \textbf{Step 2: Select Initial Interim Sample Size} \\
Set a provisional interim analysis point as a proportion (e.g., 70\%) of \deleted[id=HP]{nmax} \added[id=HP]{the preliminary sample size}.

\State \textbf{Step 3: Optimize Design Boundaries} \\
Determine the decision boundaries to control the average type I error rate and maximize average power. \added[id=HP]{This can be implemented using the procedure described in Algorithm~\ref{A11}.}

\State \textbf{Step 4: Refine Interim Sample Size} \\
Given the current \deleted[id=HP]{nmax} \added[id=HP]{maximum sample size} and boundaries, search for an interim sample size that minimizes Equation~\eqref{samplesize}.

\State \textbf{Step 5: Evaluate Power} \\
Obtain the power of the design via simulation. If the power is below the target level, proceed to Step 6.

\State \textbf{Step 6: Iterative Adjustment Loop}
\Repeat
    \State Increase \deleted[id=HP]{nmax} \added[id=HP]{maximum sample size} by a fixed increment (e.g., 5);
    \State Re-optimize the decision boundaries;
    \State Recalculate power;
\Until{target power is achieved}

\State \textbf{Step 7: Output Final Design} \\
Return the interim sample size, maximum sample size, and optimal design parameters.

\end{algorithmic}
\end{algorithm}

\added[id=HP]{From Algorithm~\ref{Al2}, two optimization strategies are considered for determining the final sample size and decision boundaries: optimal and suboptimal.}

\added[id=HP]{\textbf{Optimal Strategy}: Under the optimal method, the design parameters -- including the interim and maximum sample sizes as well as the Bayesian decision boundaries -- are iteratively optimized. At each step where the maximum sample size is incremented, the decision boundaries (e.g., posterior cutoff thresholds such as $C(n_r)$) are re-estimated by searching over a predefined grid of tuning parameters (e.g., $(\lambda, \gamma)$ in Algorithm~\ref{A11}). This ensures that the decision rule remains optimal with respect to the updated trial configuration.}

\added[id=HP]{This approach guarantees that the final design is fully optimized in both dimensions--sample size and decision thresholds--often resulting in a more efficient trial with smaller expected sample size under $H_0$ and improved power under $H_1$. However, this strategy is computationally intensive due to the repeated boundary optimization required at each candidate value of the maximum sample size.}

\paragraph{\added[id=HP]{Pragmatic Strategy \deleted[id=HP]{Suboptimal Strategy")}}} \added[id=HP]{Under the pragmatic method, boundary optimization is performed only once, based on an initial estimate of the sample size. As the sample size is adjusted iteratively to achieve the desired statistical power, the originally estimated decision boundaries are retained throughout. Only after the sample size satisfies the power requirement are the boundaries re-optimized one final time. This strategy substantially reduces computational burden while still producing a design that performs comparably to the fully optimized approach. The trade-off lies in potentially slightly lower efficiency compared to the optimal strategy, but the reduced computational cost makes it particularly appealing for use in practice.}

\added[id=HP]{These two strategies represent a balance between computational feasibility and statistical efficiency, allowing users to choose based on their design priorities and available resources.}

\deleted[id=HP]{The procedure described above is implemented in the function \texttt{Two\_stage\_sample\_size}($\cdot$) from R package \texttt{DTEBOP2}. The optimal and suboptimal strategies can be specified by setting \texttt{method = "optimal"} and \texttt{method = "suboptimal"}, respectively.}
\added[id=HP]{The procedure described above is implemented in the function \texttt{Two\_stage\_sample\_size}($\cdot$) from the R package \texttt{DTEBOP2}. The two strategies described earlier can be selected by setting \texttt{method = "optimal"} for the fully optimized approach and \texttt{method = "suboptimal"} for the pragmatic strategy. Although labeled as "suboptimal" in the function interface, this setting corresponds to the pragmatic strategy discussed above.}

\paragraph{\added[id=HP]{Design Optimization with Early Stopping Constraint}}

\deleted[id=HP]{In some situations, the probability of early stopping under the null hypothesis, denoted $PS_{H_0}$, may be undesirably low (e.g., 25\%) when computed directly by the algorithm.}
\added[id=HP]{In some applications, the default algorithm may yield a relatively low probability of early stopping under the null hypothesis (e.g., $PS_{H_0} = 25\%$),}
\added[id=HP]{which may be undesirable in phase II trials where ethical considerations and resource efficiency are important.}
\deleted[id=HP]{This reflects a limited capacity of the design to terminate early for futility, which is generally suboptimal in phase II trials where early stopping is important for ethical and resource considerations.}

\added[id=HP]{To address this issue, our method allows users to impose a lower bound constraint on $PS_{H_0}$.}\deleted[id=HP]{To address this limitation, our method allows users to specify a minimum acceptable threshold for $PS_{H_0}$.}\deleted[id=HP]{When such a threshold is provided, the algorithm no longer minimizes the expected sample size $EN$ in Equation~\eqref{samplesize}.}
\added[id=HP]{When this constraint is specified, the algorithm no longer seeks to minimize the expected sample size $EN$ defined in Equation~\eqref{samplesize}.}\deleted[id=HP]{Instead, it identifies the interim and final stage sample sizes that satisfy both the power requirement (i.e., at least $1 - \beta$ when $S \in [L,U]$) and the constraint that $PS_{H_0}$ is not below the user-specified threshold.}\added[id=HP]{Instead, it searches for the interim and final sample sizes that jointly satisfy (i) the desired power requirement (i.e., at least $1 - \beta$ when $S \in [L, U]$), and (ii) the early stopping probability under the null hypothesis exceeding the user-defined minimum.}\deleted[id=HP]{Since $PS_{H_0}$ generally increases with the interim sample size $n_1$ for a fixed total sample size $n$, the algorithm restricts the ratio $n_1/n$ to be less than 0.75 by default.}
\added[id=HP]{Because $PS_{H_0}$ generally increases with the interim sample size $n_1$ (for fixed $n$), the algorithm restricts the ratio $n_1/n$ to be below 0.75 by default.} This limit prevents designs that allocate an excessively high proportion of the total sample to the interim stage, which could reduce the remaining sample size and thereby compromise the overall statistical power.

\deleted[id=HP]{This functionality is implemented in the function \texttt{Two\_stage\_sample\_size}($\cdot$). For example, if the user wishes to ensure $PS_{H_0} \geq 0.4$, the sample size can be obtained by setting \texttt{earlystop\_prob = 0.4} in \texttt{Two\_stage\_sample\_size}($\cdot$).}
\added[id=HP]{This functionality is implemented in the function \texttt{Two\_stage\_sample\_size}($\cdot$). For example, to ensure that $PS_{H_0} \geq 0.4$, the user may set \texttt{earlystop\_prob = 0.4} in the function call.}

\section{Real Data Example}
\subsection{Example with R Code Demonstration}
\added[id=HP]{In this section, we demonstrate the software implementation of the DTE-BOP2 design using data from a real clinical trial.}

\added[id=HP]{We reconsider the CheckMate 017 study, a randomized, open-label, international phase 3 trial that evaluated the efficacy of nivolumab—a PD-1 immune checkpoint inhibitor—compared to docetaxel in patients with squamous-cell non-small-cell lung cancer (NSCLC) after prior platinum-based chemotherapy. In this example, we focus on the progression-free survival (PFS) endpoint, which exhibited delayed treatment effects.}

\added[id=HP]{To illustrate the DTE-BOP2 framework, we use the \texttt{DTEBOP2} R package to redesign the trial with PFS as the primary endpoint. The Kaplan-Meier (KM) curves were reconstructed using the \texttt{PlotDigitizer} software and the \texttt{reconstructKM} R package. The curves for docetaxel and nivolumab are nearly identical up to approximately 2.5 months before they begin to separate, indicating a delayed treatment effect. According to the original publication, the hazard ratio of nivolumab versus docetaxel was 0.62 (95\% CI, 0.47–0.81), showing a statistically significant benefit for nivolumab.}

\added[id=HP]{Below are the steps to implement DTE-BOP2 in practice:}

\begin{enumerate}
    \item[Step 1.] \added[id=HP]{Specify the overall median survival times under standard-of-care and treatment arms, denoted as $\bar{\mu}_0$ and $\bar{\mu}_1$, respectively. These values define the null and alternative hypotheses. For instance, 2.8 months vs. 3.5 months.}
    
    \item[Step 2.] \added[id=HP]{Define the prior distribution for the delay time \(S\) using a truncated Gamma distribution:}
    \[
    \added[id=HP]{S \sim \text{Gamma}(s_1, s_2)\,I(L, U),}
    \]
    \added[id=HP]{and specify a best-guess value for the most likely delay time.}
    
    \item[Step 3.] \added[id=HP]{Determine sample size and optimal parameters:}
    \begin{itemize}
        \item \added[id=HP]{If the study is a two-stage design, use the function \texttt{Two\_stage\_sample\_size($\cdot$)} to compute interim and final sample sizes, as well as the optimal tuning parameters \(\lambda\) and \(\gamma\), such that Equation~\eqref{averagetypei} is controlled below $\alpha$ and Equation~\eqref{averagepower} is maximized.}
        \item \added[id=HP]{If the study involves more than two stages, the user must provide interim and final stage sample sizes. Then use \texttt{get.optimal\_2arm\_piecewise($\cdot$)} to estimate \(\lambda\) and \(\gamma\) that optimize the design.}
    \end{itemize}
    
    \item[Step 4.] \added[id=HP]{Once data become available, the \texttt{conduct($\cdot$)} function can be used to perform interim or final analyses based on the design parameters obtained from Step 3.}
\end{enumerate}

Based on published results, we assume the \textit{overall }median progression-free survival times to be $\bar{\mu}_0 = 2.8$ months for the standard of care (SOC) arm and $\bar{\mu}_1 = 3.5$ months for the experimental arm. Suppose that type I error is controlled at $\alpha=10\%$ and type II error is controlled at $\beta=15\%$ (e.g., power is 85\%). Applying the conventional log-rank test, which does not account for potential delayed treatment effects,  the required number of events is calculated as follows:

\[
\frac{4(z_{0.95} + z_{0.85})^2}{\left(\log\left(\frac{2.8}{3.5}\right)\right)^2} = 578,
\]

\added[id=HP]{This required events is relatively large and can not be afforded in real practice.}

\added[id=HP]{In this example, we fix the follow-up window at 6 months. The interim analysis will be conducted once the target sample size is accrued in both treatment and control arms, and the final analysis will occur after the last patient has completed the 6-month follow-up. We assume a constant accrual rate of 6 patients per month per arm.}

\paragraph{\added[id=HP]{Specifying the Prior for Delayed Separation Timepoint}}

\added[id=HP]{To specify the prior distribution for the delayed separation timepoint \( S \), we incorporate expert opinion. Based on clinical experience and supporting literature, suppose that the clinical team identified a plausible range with a lower bound \( L = 2 \) months and an upper bound \( U = 2.5 \) months.}

\added[id=HP]{When experts can provide more specific estimates for \( S \), such as mean and median values, these can be used to fit a truncated Gamma prior distribution. For instance:}
\begin{itemize}
    \item \added[id=HP]{Expert A: Mean = 2.2, Median = 2.27}
    \item \added[id=HP]{Expert B: Mean = 2.1, Median = 2.3}
    \item \added[id=HP]{Expert C: Mean = 2.3, Median = 2.31}
\end{itemize}

\added[id=HP]{We assume \( S \sim \text{Gamma}(s_1, s_2)I(L, U) \), and use the function \texttt{trunc\_gamma\_para()} in the \texttt{DTEBOP2} package to estimate the parameters \( s_1 \) and \( s_2 \). In the example below, \texttt{weights = c(4,4,2,1,1)} assigns greater emphasis to the mean and median while the other three positions correspond to standard deviation, 25th percentile, and 97.5th percentile, which are set to \texttt{NULL} in this case.}

{\small
\begin{lstlisting}[language=R]
> library(DTEBOP2)
> expert_data_correct <- list(
  list(mean = 2.2, median = 2.27, sd = NULL, q25 = NULL, q975 = NULL), 
  list(mean = 2.1, median = 2.3, sd = NULL, q25 = NULL, q975 = NULL),  
  list(mean = 2.3, median = 2.31, sd = NULL, q25 = NULL, q975 = NULL) 
 )
> param <- trunc_gamma_para(L = 2, U = 2.5, expert_data = expert_data_correct,
weights = c(4,4,2,1,1), num_cores = 4) 
> print(param)
> $shape
> [1] 12.85641
> $scale
> [1] 0.1931649
\end{lstlisting}
}

\added[id=HP]{The estimated shape and scale parameters are \( s_1 = 12.86 \) and \( s_2 = 0.19 \), respectively. Figure~\ref{fig-1} displays the prior density plot, which is approximately uniform over \([2, 2.5]\), indicating strong agreement among experts.}

\begin{figure}[!htp]
    \centering
    \includegraphics[width=1\linewidth]{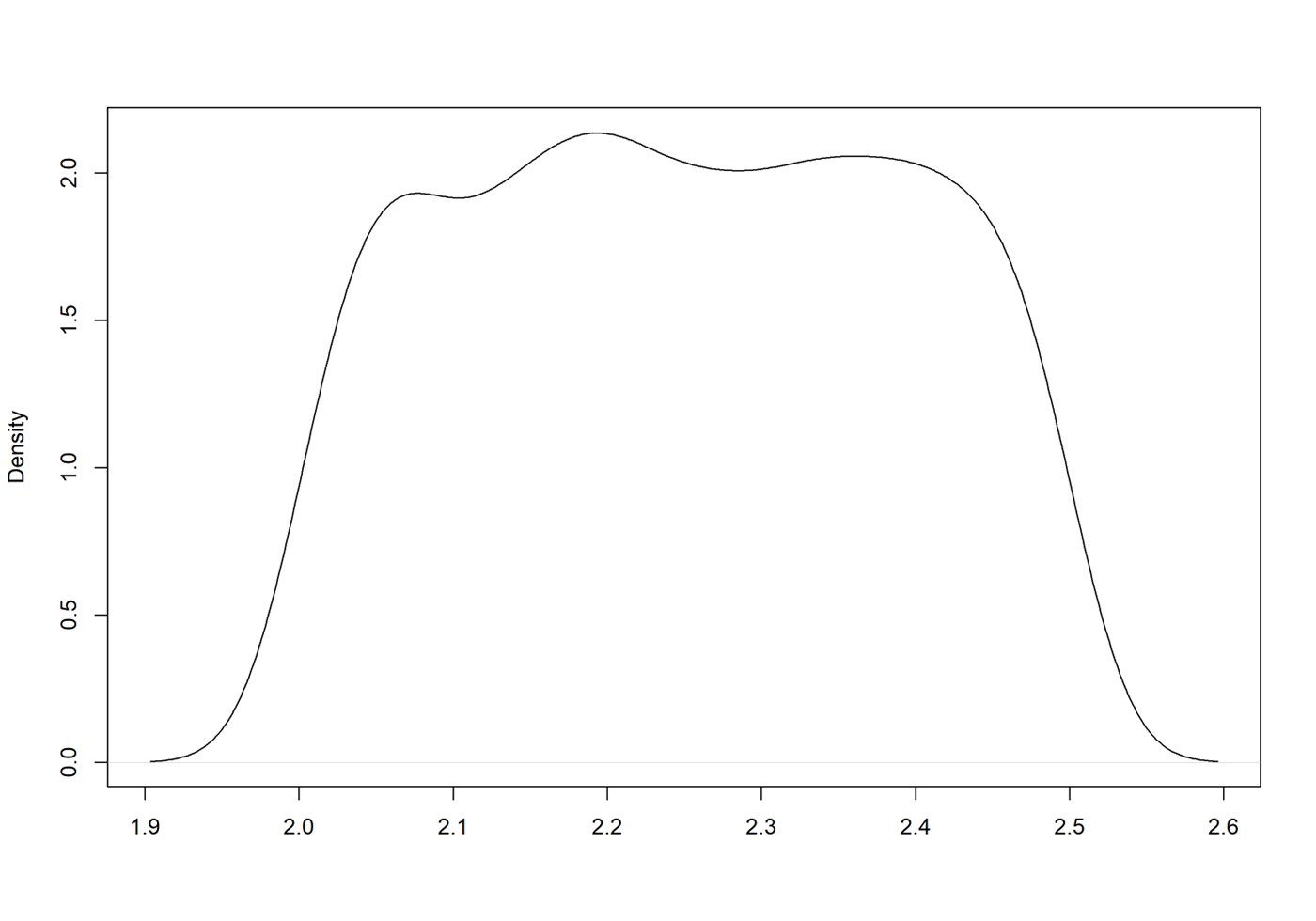}
    \caption{The density plot of $S\sim \text{Gamma}(12.86,0.19)I(2,2.5)$}
    \label{fig-1}
\end{figure}

\paragraph{\added[id=HP]{Identifying the Best-Guessed Separation Time Point}}

\added[id=HP]{For the best-guessed most likely value of the separation timepoint \( S \), denoted \( S_{\text{likely}} \), a simple default strategy is to use the mean of the truncated Gamma prior or the midpoint of the plausible range \([L, U]\), that is, \( (L+U)/2 \).}

\added[id=HP]{Alternatively, when historical data are available that resemble the current study population and setting, they can be used to inform the choice of \( S_{\text{likely}} \). In such cases, we propose a convenient and practical approach: identify the timepoint at which the Kaplan-Meier curves for the experimental and control arms begin to diverge significantly. This divergence often signals the onset of delayed treatment effect.}

\added[id=HP]{In our motivating example based on the CheckMate 017 study, we treat the published data as a form of historical dataset. By visually inspecting the reconstructed Kaplan-Meier curves, we identify that the two survival curves begin to noticeably separate at approximately 2.3 months. To quantify \( S_{\text{likely}} \), we provide illustrative R code that examines the timepoints at which the survival probabilities reach predefined thresholds (e.g., 0.54, 0.53, 0.52, 0.51) in each treatment arm. By comparing these time values between arms, one can verify the earliest signs of divergence.}

{\small
\begin{lstlisting}[language=R]
# Define target survival probabilities for evaluation
y_prob <- c(0.54, 0.53, 0.52, 0.51)
# Obtain summary of KM object
summary_fit <- summary(KM_fit)
arm_times <- list()
# Loop over treatment arms and find corresponding timepoints
for (i in 1:length(KM_fit$strata)) {
  time_points <- summary_fit$time[summary_fit$strata == names(KM_fit$strata)[i]]
  survival_probs <- summary_fit$surv[summary_fit$strata == names(KM_fit$strata)[i]]
  
  # Identify closest timepoints for each target probability
  closest_times <- sapply(y_prob, function(p) {
    closest_index <- which.min(abs(survival_probs - p))
    return(time_points[closest_index])
  })
  
  arm_times[[names(KM_fit$strata)[i]]] <- setNames(closest_times, y_prob)
}
print(arm_times)
$`arm=0`
   0.54    0.53    0.52    0.51 
 2.28803 2.48276 2.48276 2.48276   
$`arm=1`
   0.54    0.53    0.52    0.51 
 2.28923 2.54769 3.26769 3.26769
\end{lstlisting}
}

\added[id=HP]{In our CheckMate 017 example, the output shows the time at which survival probabilities reach certain levels for each arm. Based on the difference in timepoints between arms, the separation is observed to begin around 2.28 months, which is used as \(S_{\text{likely}}\) in our design.}

 \deleted[id=HP]{With $S_{\text{likely}}= 2.28$, we have $\tilde{\mu}_0=2.8$ and $\tilde{\mu}_1=6.57$ by \eqref{relation}. For prior specification of $\mu_0$ and $\mu_1$, we recommend assigning a non-informative inverse-gamma prior to $\mu_0$ centered at the null hypothesis value, $\mu_0 = 2.8/\log(2)$. Specifically, we set the prior as $\mu_0 \sim \text{Inv-Gamma}(4, 12.12)$, which yields a prior median event time of $12.12/(4-1) = 2.8/\log(2)$ and a variance of 8.16. For $\mu_1$, we assume $\mu_1 \sim \text{Inv-Gamma}(4, 24.24)$ so that its prior mean is twice that of $\mu_0$ under the null hypothesis. These priors are used as the default settings in our software package \texttt{DTEBOP2} when the user sets the priors for $\mu_0$ and $\mu_1$ to \texttt{null}. Based on extensive simulation studies, this default prior specification offers stable and reliable performance across a range of scenarios.}

 \paragraph{\added[id=HP]{Prior Specification for \boldmath$\mu_0$ and $\mu_1$}}

\added[id=HP]{Given the best-guess separation timepoint $S_{\text{likely}} = 2.28$ months, Equation~\eqref{relation} yields the median survival times $\tilde{\mu}_0 = 2.8$ and $\tilde{\mu}_1 = 6.57$ months.}

\added[id=HP]{Following the recommendation of \citet{Zhou:2020}, we assign weakly-informative inverse-gamma priors centred at the null values.  Specifically, we set}
\[
\added[id=HP]{\mu_0 \sim \operatorname{Inv\text{-}Gamma}(a_0=4,\; b_0=12.12), \qquad
  \mu_1 \sim \operatorname{Inv\text{-}Gamma}(a_1=4,\; b_1=24.24).}
\]

\added[id=HP]{With shape parameter $a=4$, the prior mean is $b/(a-1)$.  Hence}
\[
\added[id=HP]{\text{Mean}(\mu_0)=\frac{12.12}{3} = 4.04\;\text{months}, 
  \qquad
  \text{Mean}(\mu_1)=\frac{24.24}{3} = 8.08\;\text{months},}
\]
\added[id=HP]{so the prior for $\mu_1$ is centred at twice the prior mean of $\mu_0$, reflecting the anticipated treatment benefit.  The prior variance of $\mu_0$ is $b_0^2/[(a_0-1)^2(a_0-2)]=8.16$, giving a relatively diffuse (non-informative) spread.}

\added[id=HP]{These hyper-parameters are used as the \texttt{DTEBOP2} default when the user leaves the \texttt{gprior.E\_1} and \texttt{gprior.E\_2} arguments set to \texttt{NULL}.  Extensive simulation indicates that this default choice yields stable operating characteristics across a wide range of scenarios, while allowing advanced users to input more informative priors if desired.}

\deleted[id=HP]{Next, given the hypothesis of median survival times of 3.5 months for the experimental arm and 2.8 months for the SOC arm, and from the above with the best guessed delayed separation time point of 2.28 months within the range of $L = 2$ months to $U = 2.5$ months, by assuming accrual rate with 6 patients per month per arm, follow-up of 6 months for the last patient enrolled at final stage. The estimated sample sizes for interim and final stages and the optimal parameters $\lambda,\gamma$ can be computed as follows using function.} 

\deleted[id=HP]{The planned interim and final analyses correspond to per-arm sample sizes of 28 and 40 patients, respectively. Accordingly, the first interim analysis is anticipated to occur approximately 4.67 months after the initiation of recruitment, assuming a steady accrual rate. The optimal parameters, $\lambda$ and $\gamma$, are 0.95 and 1, respectively, with an average type I error rate of 0.0871 and power of 0.8667. Meanwhile, if we may be interested in the corresponding expected number of events for two stages, we can use function.}

\deleted[id=HP]{We anticipate observing approximately 23 events at the interim analysis stage and a total of 68 events by the final analysis. }

\paragraph{\added[id=HP]{Designing a Two-Stage Study with \texttt{DTEBOP2}}}

\added[id=HP]{We now combine the inputs established above to obtain a two-stage DTE-BOP2 design.  
The working hypotheses are an overall median PFS of $\bar{\mu}_0 = 2.8$ months for SOC and $\bar{\mu}_1 = 3.5$ months for the experimental arm.  
The best-guess separation timepoint is $S_{\text{likely}} = 2.28$ months, within the expert-elicited range $L = 2$ to $U = 2.5$ months.  
We further assume a constant accrual rate of 6 patients per month per arm and a fixed follow-up of 6 months for the last-enrolled patient.  
Type I and II error thresholds are set to $\alpha = 0.10$ and $\beta = 0.15$ (power = 85\%).}

{\small
\begin{lstlisting}[language=R]
library(DTEBOP2)

## --- User inputs -------------------------------------------------
median.1   <- 2.8      # SOC median (bar mu_0)
median.2   <- 3.5      # EXP median (bar mu_1)
L          <- 2        # lower bound for S
U          <- 2.5      # upper bound for S
S_likely   <- 2.28     # best-guess separation time
trunc.para <- c(12.86, 0.19)  # shape & scale of truncated Gamma prior
rate       <- 6        # accrual per month per arm
FUP        <- 6        # follow-up (months) for last patient
err1       <- 0.10     # type I error
err2       <- 0.15     # type II error
## -----------------------------------------------------------------

Two_stage_sample_size(median.1 = median.1, median.2 = median.2,
  L = L, U = U, S_likely = S_likely,
  err1 = err1, err2 = err2, trunc.para = trunc.para,
  FUP = FUP, rate = rate, weight = 0.5,
  earlystop_prob = NULL, seed = 123)
#> $sample_size   # per-arm n at interim and final
#> [1] 28 40
#> $optimal       # lambda, gamma, avg type I error, avg power
#> [1] 0.9500 1.0000 0.0871 0.8667
\end{lstlisting}
}

\added[id=HP]{The function returns a two-stage schedule with 28 patients per arm at the interim look and 40 per arm at the final analysis.  
With an accrual rate of 6 patients/month/arm, the first analysis is expected ≈ 4.7 months after study start.  
The optimal decision parameters are $\lambda = 0.95$ and $\gamma = 1.0$, giving an average type I error of 0.087 and power of 0.867 under the assumed distribution of $S$.}

\added[id=HP]{If one wishes to know the corresponding expected \emph{number of events} at each stage, the helper function \texttt{event\_fun()} can be used:}

{\small
\begin{lstlisting}[language=R]
event_fun(median.1 = 2.8, median.2 = 3.5, S_likely = 2.28,
          n.interim = c(28, 40), rate = 6, FUP = 6,
          n.sim = 1000, seed = 123)
#>             Events
#> 1st interim     23
#> Final stage     68
\end{lstlisting}
}

\added[id=HP]{Thus, we expect roughly 23 PFS events at the interim analysis and a total of 68 events by the final analysis under the assumed accrual and follow-up scheme.}

\subsection{Sensitivity Analyses}

\paragraph{\added[id=HP]{Sensitivity Analysis I: Varying the True Separation Time \(S\)}}

\added[id=HP]{We keep the design fixed at
\(S_{\text{likely}}=2.28\) months,
\(\tilde{\mu}_0=2.8\) months,
\(\tilde{\mu}_1=6.57\) months,
\((n_1,n)=(28,40)\) per arm,
and the optimal decision parameters \(\lambda=0.95,\;\gamma=1.0\).
We vary the \emph{true} but unknown separation time \(S\in\{2.0,2.1,\dots,2.5\}\)
to evaluate the following operating characteristics:} (i) \emph{percentage of rejecting the null} (PRN), which is type I error under
\(H_0\) or power under \(H_1\);
(ii) \emph{percentage of early termination} (PET);
(iii) average sample size (each arm);
(iv) average trial duration (months).

\added[id=HP]{The function \texttt{getoc\_2arm\_piecewise()} computes these OCs.
Below is an example for \(S=2.0\) months under \(H_0\):}

{\small
\begin{lstlisting}[language=R]
getoc_2arm_piecewise(
  median.true  = c(2.8, 2.8),    # H0
  lambda       = 0.95,
  gamma        = 1,
  n.interim    = c(28, 40),
  L = 2, U = 2,                  # true S = 2.0
  S_likely     = 2.28,
  FUP          = 6,
  trunc.para   = c(12.86, 0.19),
  rate         = 6,
  nsim         = 1e4, seed = 10)
#> $earlystop        0.2957
#> $reject           0.0841
#> $average.patients 36.45
#> $trial.duration   10.34
\end{lstlisting}
}

\added[id=HP]{The full results for all six values of \(S\) are summarised in Table \ref{tab:Sens1}.}

\begin{table}[H]
\centering
\caption{Operating characteristics of the two-stage DTE-BOP2 design
for six true separation times \(S\) (units: months). Sample size is the per-arm expected accrual patients and trial duration is reported in months.}
\label{tab:Sens1}
\resizebox{\textwidth}{!}{
\begin{tabular}{|c|c|c|c|c|c|}
\hline
True \(S\) & Scenario & PRN (\%) & PET (\%) & Sample size & Duration$^{\dagger}$ \\ \hline
\multirow{2}{*}{2.0} & \(H_0\) & 8.4 & 29.6 & 36.5 & 10.3 \\ \cline{2-6}
                     & \(H_1\) & 88.0 & 7.1 & 39.1 & 12.1 \\ \hline
\multirow{2}{*}{2.1} & \(H_0\) & 8.9 & 28.3 & 36.6 & 10.4 \\ \cline{2-6}
                     & \(H_1\) & 87.2 & 7.4 & 39.1 & 12.1 \\ \hline
\multirow{2}{*}{2.2} & \(H_0\) & 8.7 & 27.1 & 36.7 & 10.5 \\ \cline{2-6}
                     & \(H_1\) & 87.1 & 7.1 & 39.1 & 12.1 \\ \hline
\multirow{2}{*}{2.3} & \(H_0\) & 8.7 & 26.3 & 36.8 & 10.6 \\ \cline{2-6}
                     & \(H_1\) & 86.6 & 7.1 & 39.1 & 12.1 \\ \hline
\multirow{2}{*}{2.4} & \(H_0\) & 8.8 & 25.0 & 37.0 & 10.7 \\ \cline{2-6}
                     & \(H_1\) & 86.2 & 7.3 & 39.1 & 12.1 \\ \hline
\multirow{2}{*}{2.5} & \(H_0\) & 8.7 & 23.8 & 37.1 & 10.8 \\ \cline{2-6}
                     & \(H_1\) & 85.5 & 7.6 & 39.1 & 12.1 \\ \hline
\end{tabular}
}
\begin{flushleft}
\footnotesize $^{\dagger}$Months from first patient in to final analysis.
\end{flushleft}
\end{table}

\added[id=HP]{The simulation confirms that, with \(S_{\text{likely}}=2.28\) fixed, the average type I error stays below 0.10 for all six true \(S\) values, and the power remains above 0.85.  
As the true separation moves from 2.0 to 2.5 months the power gradually declines, because a larger portion of non-separated follow-up time is included in the test statistic.  
Conversely, when the true separation occurs earlier than \(S_{\text{likely}}\), the power increases modestly.}

\added[id=HP]{To more comprehensively assess the DTE-BOP2 design we evaluated four supplementary scenarios:  
(i) enforcing a minimum early-stopping probability of 0.40 under \(H_0\);  
(ii) exploring alternative priors for \(S\) and extending to a three-stage design;  
(iii) varying the accrual rate (slow--fast); and  
(iv) varying the user-specified \(S_{\text{likely}}\) within \([2.0,2.5]\) instead of fixing it at 2.28 months.  
Detailed numerical results for all four scenarios appear in the on-line Supplementary Material.  
Below we highlight the main findings for the last two scenarios.}

\deleted[id=HP]{To more comprehensively assess the performance of the proposed DTE-BOP2 design, we conducted a series of supplementary simulation studies under alternative scenarios. These include: (i) enforcing a minimum early stopping probability of 0.4 under the null hypothesis; (ii) evaluating the impact of alternative prior specifications for $S$ and extending the design to a three-stage framework; (iii) examining the operating characteristics of the design under slower and faster accrual rates; and (iv) varying the user-specified best-guess separation time $S_{\text{likely}}$ within the interval $[2.0, 2.5]$ rather than fixing it at 2.28. Due to space constraints, detailed results are presented in the Supplementary Material. Here, we summarize the key findings from the last two scenarios:}

\paragraph{Sensitivity Analysis II: Impact of Accrual Rate}

\deleted[id=HP]{Table~4 in the Supplementary Material summarizes the impact of varying recruitment rates (3, 6, 9, and 12 patients per arm per month) on key operating characteristics. The findings reveal several important trends:}
\added[id=HP]{Table 4 (Supplement) summarises OCs for recruitment rates of 3, 6, 9 and 12 patients/arm/month.  Key observations are:}
\begin{itemize}
    \item As the accrual rate increases, the early stopping probability under $H_0$ (PET($H_0$)) decreases substantially. For example, PET($H_0$) drops from 40.9\% at rate = 3 to 14.0\% at rate = 12 when $S = 2$.
    \item Type I error remains stable across all accrual rates, consistently controlled within 8.2\% to 9.3\%.
    \item Power declines as the accrual rate increases. Compared to the rate of 6 (used in the main analysis), power decreases by approximately 3\% at rate = 9 and 6\% at rate = 12 (i.e., doubling the accrual rate), primarily due to reduced follow-up time for the delayed treatment effect to manifest.
    
    \item The gap between PET($H_0$) and PET($H_1$) narrows markedly with increasing accrual rates, from over 35\% at rate = 3 to below 4\% at rate = 12. This convergence reduces the discriminative utility of early stopping, especially when the treatment effect emerges late in follow-up \citep{Liu:Takeda:2025, Miki:Uno:2019}.
\end{itemize}

Overally, these results suggest that the design maintains favorable operating characteristics when the actual accrual rate does not deviate substantially from the planned rate.

\paragraph{\added[id=HP]{Sensitivity Analysis III: Varying $S_{\text{likely}}$}}

\added[id=HP]{We replaced the best-guessed \(S_{\text{likely}}=2.28\) with three values (2.00, 2.25, 2.50 months) and re-calculated the two-stage sample size based on \(S_{\text{likely}}=2.00\). The required per-arm sizes increase to (44, 63) with \(S_{\text{likely}}=2.00\).  Table~\ref{Table_s} summarises OCs across all true \(S\) and design-stage guesses:}
\begin{enumerate}
  \item \added[id=HP]{With $n_{\max}=63$ per arm, the design controls the type I error at $\leq 0.10$ and achieves power $\geq 0.85$ for all combinations of the true $S$ and $S_{\text{likely}}$ within $[2.0, 2.5]$. Thus, under the most conservative sample size strategy, the design is doubly robust to misspecification of $S_{\text{likely}}$ and the true unknown $S$.}
  \item \added[id=HP]{A wider prior range \([L,U]\) inflates sample size, because power must be preserved even when the true separation is close to \(U\).}
  \item \added[id=HP]{Choice of \(L,U\) should therefore reflect genuine expert knowledge; overly wide intervals exact a cost in sample size, echoing the recommendation of \citet{Ye:2019}.}
\end{enumerate}

\begin{table}[H]
\centering
\caption{Type I error, power, average sample size under $H_1$, and trial duration under $H_1$ with sample size (44, 63) per arm.}
\resizebox{\textwidth}{!}{%
\begin{tabular}{|c|c|c|c|c|c|}
\hline
\multicolumn{2}{|c|}{\textbf{Scenario}} & \textbf{Type I Error(\%)} & \textbf{Power(\%)} & \textbf{Average Sample Size (H1)} & \textbf{Duration$^{\dagger}$ (H1)} \\
\hline
\multirow{6}{*}{$S_{\text{likely}} = 2$} 
& $S = 2$     & 7.8 & 89.1 & 61.9 & 16.0 \\
& $S = 2.1$   &  8.1 & 88.3 & 61.9 & 16.0 \\
& $S = 2.2$   &  8.3 & 88.3 & 61.8 & 15.9 \\
& $S = 2.3$   &  8.3 & 87.6 & 61.8 & 15.9 \\
& $S = 2.4$   &  8.0 & 87.3 & 61.8 & 15.9 \\
& $S = 2.5$     & 8.2 & 86.9 & 61.8 & 15.9 \\
\hline
\multirow{6}{*}{$S_{\text{likely}} = 2.25$} 
& $S = 2$   & 7.8 & 96.4 & 62.5&16.3  \\
& $S = 2.1$   &8.1 & 96.2 & 62.5 & 16.3 \\
& $S = 2.2$   &8.3 & 96.0 & 62.4 & 16.2 \\
& $S = 2.3$   &  8.3 & 95.9 & 62.5 & 16.2 \\
& $S = 2.4$   &  8.0 & 95.5 & 62.4 & 16.2 \\
& $S = 2.5$   & 8.2 & 95.3 & 62.4 & 16.2 \\
\hline
\multirow{6}{*}{$S_{\text{likely}} = 2.5$} 
& $S = 2$   & 7.8 & 99.5 & 62.9 & 16.4 \\
& $S = 2.1$   &  8.1 & 99.6 & 62.9 & 16.5 \\
& $S = 2.2$   &  8.3 & 99.5 & 62.9 & 16.4 \\
& $S = 2.3$   &  8.3 & 99.3 & 62.9 & 16.4 \\
& $S = 2.4$   & 8.3 & 99.1 & 62.8 & 16.4 \\
& $S = 2.5$   & 8.2& 99.1 & 62.8 &  16.4\\
\hline
\end{tabular}
}
\begin{flushleft}
\footnotesize $^{\dagger}$Months from first patient in to final analysis.
\end{flushleft}
\label{Table_s}
\end{table}

\subsection{Comparison with Frequentist Approaches Incorporating Delayed Treatment Effects}
\added[id=HP]{We reuse the CheckMate~017 dataset introduced in previous section as our real-data template for all simulations in this subsection.}
\deleted[id=HP]{For a comparative study, we adopt parameters from a real data example. Specifically, we set $L=2$, $U=2.5$, $\bar{\mu}_0 = 2.8$ months and $\bar{\mu}_1 = 3.5$ months for the SOC and experimental arms, respectively. We assume the most likely separation timepoint of $S_{\text{likely}} = 2.28$ months for DTE-BOP2. Accordingly, the hazard rates are $\log(2)/2.8$ for both arms with experimental arm before $S$, and $\log(2)/6.57$ for the experimental arm after $S$. We also assume a 6-month follow-up period and an accrual rate of 6 patients per month per arm.}

The log-rank test statistic, commonly used in survival analysis, is given by:

\begin{equation}
 \frac{\sum_{j \in D} (d_{1j} - e_{1j})}{\sqrt{\sum_{j \in D} \frac{n_{1j} n_{2j} d_j (n_j - d_j)}{n_j^2(n_j - 1)}}},
\label{log-rank}
\end{equation}
where $D$ denotes the set of indices corresponding to observed event times. At each event time $t_j$, $d_{1j}$ and $d_{2j}$ represent the number of events in the standard of care (SOC) and experimental groups, respectively. Likewise, $n_{1j}$ and $n_{2j}$ indicate the number of patients at risk in each group just before $t_j$. Here, $d_j = d_{1j} + d_{2j}$, $n_j = n_{1j} + n_{2j}$, and the expected number of events in the SOC group under the null hypothesis is $e_{1j} = n_{1j} d_j / n_j$.

To account for delayed treatment effects (DTE), a piecewise weighted log-rank test has been proposed \citep{Xu:2017, Wu:2019}, given by:

\begin{equation}
L = \frac{\sum_{j \in D_2} (d_{1j} - e_{1j})}
{\left\{ \sum_{j \in D_2} \frac{n_{1j} n_{2j} d_j (n_j - d_j)}{n_j^2(n_j - 1)} \right\}^{1/2}},
\label{Wlog-rank}
\end{equation}
where $D_2$ includes only those event times that occur after the separation timepoint $S$.

\added[id=HP]{Throughout this subsection, the clinical (overall) medians are fixed at $\bar{\mu}_0=2.8$ and $\bar{\mu}_1=3.5$ months, and the elicited range for the true separation time is $[L,U]=[2.0,2.5]$ months with best guess $S_{\text{likely}}=2.28$ months.  
Under the piecewise-exponential model, the experimental-arm post-separation median becomes $\tilde{\mu}_1=6.57$ months (Equation~\eqref{relation}), giving hazard rates} 
\[
\added[id=HP]{
\lambda_0=\frac{\log 2}{2.8}\quad(\text{SOC, all time}),\qquad
\lambda_1(t)=
  \begin{cases}
    \log 2/2.8, & t<S,\\[4pt]
    \log 2/6.57, & t\ge S.
  \end{cases}}
\]
\added[id=HP]{All scenarios assume a 6-month follow-up for the last patient and an accrual rate of 6 patients per arm per month.}

\deleted[id=HP]{We compare the proposed DTE-BOP2 design with the conventional log-rank test, Wu's method \citep{Wu:2019}, and Xu's method \citep{Xu:2017}. Both Wu's and Xu's approaches are based on the piecewise weighted log-rank framework, but differ in their formula of sample size: Wu's method offers a more general sample size formula, whereas Xu's approach is tailored to local alternatives. Simulation results from \cite{Wu:2019} suggest that Wu's method is more powerful, particularly when the hazard ratio is small.}
\added[id=HP]{For comparison we consider four tests: the conventional log-rank, Xu's piecewise weighted log-rank, Wu's extension of that test, and DTE-BOP2.  
Xu's formula is derived under a local-alternative assumption, whereas Wu's handles arbitrary effect sizes; both require the design separation point \(S\) to be preset.  To create a favourable scenario for these frequentist competitors we supply them with the \emph{true} \(S\) in every simulation.  Earlier work \citep{Wu:2019} shows Wu's test to be the more powerful of the two when the hazard ratio is small.}

\deleted[id=HP]{The comparison is conducted from two perspectives:}
\added[id=HP]{The four methods are evaluated using two practical criteria:}
\begin{itemize}
    \deleted[id=HP]{\item Fixing the sample size at 40 per arm, we evaluate the empirical power under true values of $S = 2, 2.1, 2.2, 2.3, 2.4, 2.5$.}
    \item \added[id=HP]{\textbf{Fixed sample size} ($N = 40$ per arm): empirical power is estimated for true separation times $S = 2.0, 2.1, 2.2, 2.3, 2.4, 2.5$\,\text{months}.}
    \deleted[id=HP]{\item Calculating the required sample size for each method under type I and II error rates of $\alpha = 0.1$ and $\beta = 0.15$, respectively.}
    \item \added[id=HP]{\textbf{Fixed operating characteristics} ($\alpha = 0.10$, $\beta = 0.15$): the minimum per-arm sample size required by each method is determined.}
\end{itemize}

\deleted[id=HP]{To ensure consistency across methods, we consider the DTE-BOP2 design without interim analyses, since the frequentist approaches do not incorporate group sequential testing. Our results show that with $N = 40$ per arm, DTE-BOP2 maintains type I error control for $S \in [2, 2.5]$, making it valid for comparison.}
\added[id=HP]{Xu's and Wu's piece-wise weighted log-rank tests do not incorporate group-sequential monitoring, so we run DTE-BOP2 in a single-analysis mode for parity.  
With $N = 40$ per arm, DTE-BOP2 still controls the type-I error below $0.10$ for every $S \in [2, 2.5]$; its empirical power at this sample size can therefore be fairly compared with the frequentist alternatives.}

\deleted[id=HP]{To estimate empirical power, we simulate 10,000 datasets. For the experimental group, event times are generated from a piecewise exponential distribution, with hazard rates $\log(2)/2.8$ before $S$ and $\log(2)/6.57$ after $S$. SOC group events follow a standard exponential distribution with hazard rate $\log(2)/2.8$. We calculate p-values using both the log-rank test \eqref{log-rank} and the piecewise weighted log-rank test \eqref{Wlog-rank} with the true $S$, consistent with the optimal configuration described in \cite{Wu:2019}. The results are presented in Table \ref{power_comp}.}

\added[id=HP]{We generated 10\,000 simulated datasets for each true separation time $S \in \{2.0, 2.1, \dots , 2.5\}$ with a fixed sample size of $N = 40$ patients per arm.  
Events in the SOC arm followed an exponential distribution with hazard $\lambda_0=\log 2 / 2.8$.  
Events in the experimental arm followed a piece-wise exponential distribution with the same hazard before $S$ and $\lambda_1=\log 2 / 6.57$ after $S$.  
For every dataset we computed:  
(i) a conventional log-rank $p$-value \eqref{log-rank};  
(ii) a piecewise weighted log-rank $p$-value \eqref{Wlog-rank} using weight $0$ for $t<S$ and $1$ for $t\ge S$ (the ``optimal'' choice in \citet{Wu:2019}); and  
(iii) the DTE-BOP2 posterior decision.  Empirical power is summarised in Table~\ref{power_comp}.}

\begin{table}[H]
\centering
\caption{\added[id=HP]{Empirical power (10\,000 simulations, $N=40$ per arm, percentage) for three tests under six true separation times.}}
\label{power_comp}
\begin{tabular}{|l|cccccc|}
\hline
\textbf{Method} & $S=2.0$ & $2.1$ & $2.2$ & $2.3$ & $2.4$ & $2.5$ \\ \hline
DTE-BOP2        & 0.94 & 0.93 & 0.93 & 0.93 & 0.93 & 0.92 \\
PW log-rank     & 0.77 & 0.75 & 0.75 & 0.73 & 0.72 & 0.71 \\
Log-rank        & 0.51 & 0.49 & 0.48 & 0.46 & 0.42 & 0.40 \\ \hline
\end{tabular}
\end{table}

\added[id=HP]{DTE-BOP2 has 92 -- 94\% power across all $S$, exceeding the piece-wise weighted log-rank test by at least 15 percentages and the conventional log-rank by more than 40\%, even though the latter two tests are supplied with the true separation point.}

\deleted[id=HP]{For the sample size comparison, the required sample size for DTE-BOP2 is determined such that power exceeds $1 - \beta$ across the interval $S \in [L, U]$, with type I error controlled. For the log-rank test, the minimum sample size is identified through simulations ensuring empirical power exceeds $1 - \beta$. For Wu's and Xu's methods, sample sizes are computed using their respective analytical formulas, assuming the true value of $S$ is known. The total required sample sizes for each method are shown in Table \ref{size_comp}.}

\added[id=HP]{\textbf{Sample-size comparison.}  
For DTE-BOP2 we choose the per-arm sample size so that two-sided type-I error is controlled at $\alpha = 0.10$ and the power is at least $1-\beta = 0.85$ for \emph{every} $S \in [L,U]$. Because we calibrate to the worst-case $S \in [L,U]$, the required sample size for DTE-BOP2 is the same (64/arm) for all $S$.  
For the conventional log-rank test, we search via simulation and obtain the minimum $N$ that yields empirical power $\ge 0.85$ at each fixed true $S$.  
For the weighted log-rank tests we apply the analytical formulas of \citet{Xu:2017} and \citet{Wu:2019}, giving them the \emph{true} separation time, which represents an optimistic scenario for these methods.  Table~\ref{size_comp} summarises the required per-arm sample sizes.}

\begin{table}[H]
\centering
\caption{\added[id=HP]{Per-arm sample size required to achieve two-sided $\alpha = 0.10$ and power $0.85$ under six true separation times.}}
\label{size_comp}
\begin{tabular}{|l|cccccc|}
\hline
\textbf{Method} & $S=2.0$ & 2.1 & 2.2 & 2.3 & 2.4 & 2.5 \\ \hline
DTE-BOP2      & 64 & 64 & 64 & 64 & 64 & 64 \\
Wu (2019)     & 98 & 101 & 104 & 108 & 111 & 115 \\
Xu (2017)     & 97 & 100 & 103 & 107 & 110 & 114 \\
Log-rank      & 166 & 174 & 182 & 192 & 202 & 212 \\ \hline
\end{tabular}
\end{table}

\added[id=HP]{DTE-BOP2 needs only 64 patients per arm -- about 40\% of the sample required by the weighted log-rank tests and less than one-third of that required by the conventional log-rank.  
The gap widens as the true separation time approaches the upper bound of the plausible range: Xu's and Wu's designs must inflate $N$ to compensate for the longer delay, whereas DTE-BOP2 already guards against this uncertainty at the design stage.  
The efficiency gain stems from modelling the delayed effect explicitly and basing the calculation on the post-delay median ratio $\tilde{\mu}_1/\tilde{\mu}_0$ rather than the attenuated overall ratio $\bar{\mu}_1/\bar{\mu}_0$.}

%\begin{table}[!hpt]
%\centering
%\resizebox{\textwidth}{!}{%
%\caption{Percentage of rejecting the null (PRN), percentage of early termination (PET), the actual sample size (each arm), and trial duration under the DTE-BOP2 and BOP2 design with a time-to-event endpoint in a two-arm randomized trial.}
%\begin{tabular}{|c|c|c|c|c|c|}
%\hline
% &&  PRN (\%) & PET (\%) & Sample size & Trial duration\\ \hline
%\multirow{2}{*}{$S=2,S_\text{likely}=2$} 
%                       & Under $H_0$&  7.8 & 40.9 & 55.2 & 12.8 \\ \cline{2-6}
%                        & Under $H_1$&  89.4 & 5.4 & 62.0 & 16.0 \\ \hline
%\multirow{2}{*}{$S=2,S_\text{likely}=2.5$}  
%                     & Under $H_0$ &  7.8& 40.9 & 55.2 & 12.8   \\ \cline{2-6}
%                       & Under $H_1$&  99.5 & 0.5 & 62.9 & 16.4 \\ \hline
%\multirow{2}{*}{$S=2.5,S_\text{likely}=2$}
%                      & Under $H_0$ &  8.1& 37.8 & 55.8 &13.1  \\ \cline{2-6}
%                        & Under $H_1$&  86.9 & 6.5 & 61.8 & 15.9 \\ \hline
%\multirow{2}{*}{$S=2.5,S_\text{likely}=2.5$}
%                       & Under $H_0$&  8.1& 37.8 & 55.8 &13.1  \\ \cline{2-6}
%                        & Under $H_1$&  99.3 & 0.6 & 62.9 & 16.4 \\ \hline
%\end{tabular}
%\label{Table_s}
%}
%\end{table}

\begin{figure}[H]
\centering
\includegraphics[width=1.1\linewidth]{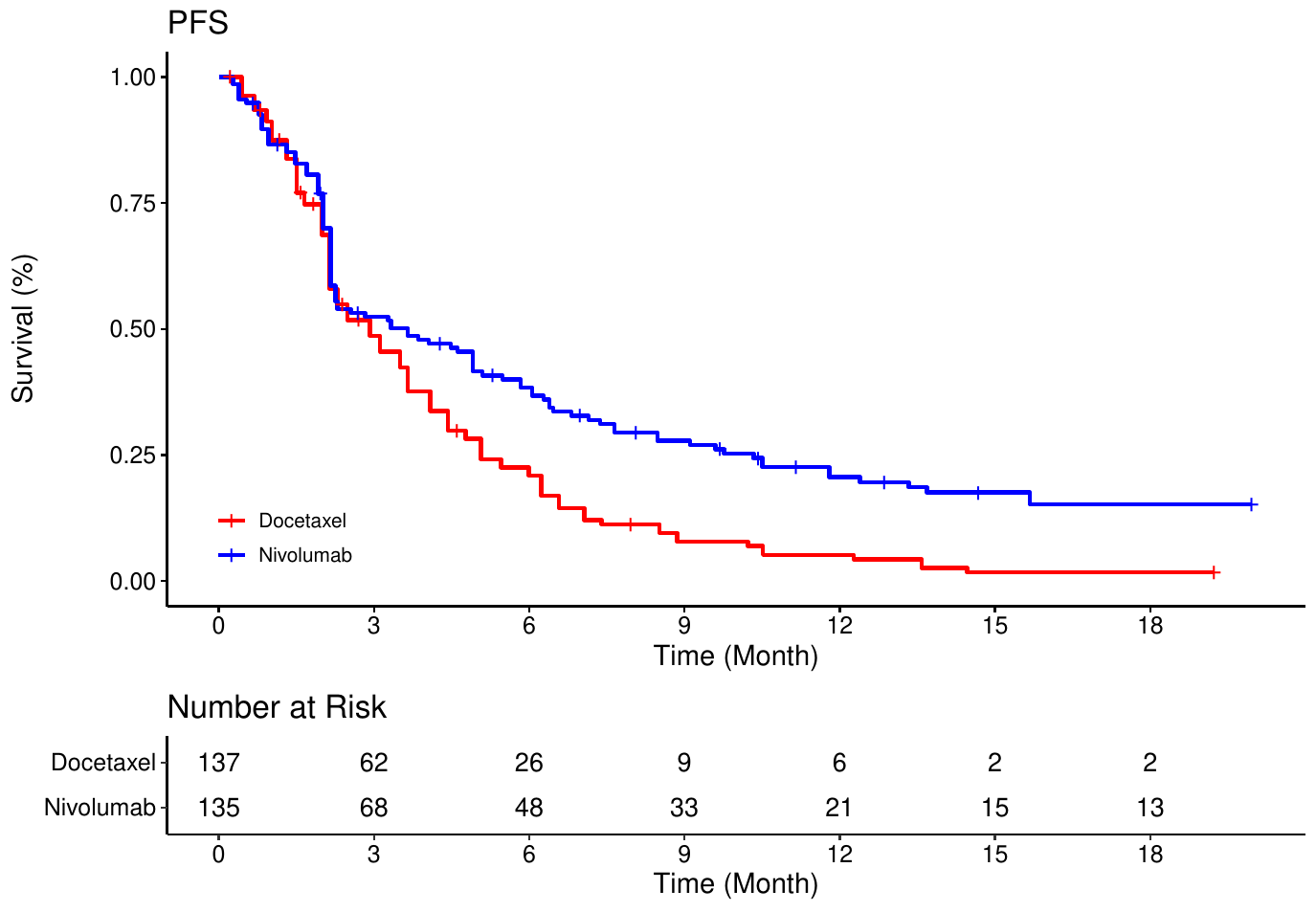}
\caption{Kaplan-Meier plot of PFS in Checkmate 017, in which DTEs are present. The control and experimental treatment curves follow the same trajectory for approximately 2.3 months, after which they separate.}
\label{real_data}
\end{figure}
 
\section{\added[id=HP]{Design Performance Evaluation via Simulation}\deleted[id=HP]{Simulation Study}}
\noindent
\added[id=HP]{While previous sections used simulation to evaluate the statistical power and sample size requirements under varying true values of $S$, this section focuses on evaluating the operating characteristics of the DTE-BOP2 design from other aspects. We assess its performance in terms of early stopping probability, trial duration, and sample size, and compare it with the original BOP2 design. Additional simulations are conducted to study type I error control across a range of $S$ and to evaluate robustness with respect to prior specification of $(L, U)$.}

\subsection{Comparison with BOP2}

We first evaluate and compare the operating characteristics of the DTE-BOP2 design against the BOP2 design \citep{Zhou:2020}. For all designs, the type I error rate was controlled at 0.1. Event times in the control arm were assumed to follow an exponential distribution, \( \text{Exp}(\mu_0) \), while those in the treatment arm were generated according to the density specified in Equation \eqref{eq2.1-1} \added[id=HP]{, which characterizes delayed treatment effects through a piecewise exponential model with change-point at separation time $S$}. Patient arrivals were modeled as a Poisson process with an intensity of \(1/3\) per month per arm, corresponding to an average recruitment rate of three patients per month per arm. \deleted[id=HP]{Following the enrollment of the last patient, follow-up continued for an additional six months.} \added[id=HP]{After the last patient was enrolled, all patients were followed for an additional six months to ensure sufficient observation of events.} Each scenario was evaluated using 10,000 simulated trials.

\deleted[id=HP]{The total sample size for each trial was fixed at 80 patients, with randomization in a 1:1 ratio between the experimental and control arms. One interim analysis was planned when 30 patients had been enrolled in each arm. Performance was assessed based on four key metrics: PRN, PET, average sample size, and trial duration, as defined in Section 3.}

\added[id=HP]{The total sample size for each simulated trial was fixed at 80 patients, with equal randomization (1{:}1) between the experimental and control arms. An interim analysis was planned after 30 patients had been enrolled in each arm. The performance of each design was evaluated using four key metrics: the probability of rejecting the null (PRN), the probability of early termination (PET), the average sample size, and the total trial duration, as defined in Section~3.}

 \deleted[id=HP]{For comparative purposes, we introduce the following concepts:}

\added[id=HP]{For clarity we distinguish between two types of simulated data used in the comparison: calibration datasets for tuning design parameters and evaluation datasets for estimating operating characteristics.}

\begin{itemize}
    \item \deleted[id=HP]{\textbf{Training datasets under the null and alternative hypotheses:}}%
          \added[id=HP]{\textbf{Calibration datasets (under $H_0$ and $H_1$):}}
    \deleted[id=HP]{Suppose the clinician wishes to test the hypotheses:}
    \added[id=HP]{To tune the design we assume a study wishes to test}
    \begin{equation}
        H_0: \bar{\mu}_0 = a \text{ months}, 
        \quad 
        H_1: \bar{\mu}_1 = b \text{ months}, 
        \quad (a < b).
    \end{equation}
    \deleted[id=HP]{Under the alternative hypothesis, training datasets for DTE-BOP2 are generated based on Equation \eqref{eq2.1-1}, with \( S \) following \eqref{eq2.1_2}. In contrast, for BOP2, datasets are generated from an exponential distribution with rate \( b/\log(2) \). Under the null hypothesis, both designs use datasets generated from an exponential distribution with rate \( a/\log(2) \).}
    \added[id=HP]{Under $H_1$, calibration datasets for \emph{DTE-BOP2} are drawn from the delayed-effect model in Equation~\eqref{eq2.1-1}, with the separation time $S$ sampled from the prior in Equation~\eqref{eq2.1_2}.  
    For \emph{BOP2} the same datasets are generated from a single-exponential distribution with median $b$ (hazard $\log 2 / b$).  
    Under $H_0$, both designs use exponential data with median $a$ (hazard $\log 2 / a$).}
    
    \item \deleted[id=HP]{\textbf{Testing datasets:} These datasets are utilized to evaluate the operating characteristics of DTE-BOP2 and BOP2, and are generated from Equation \eqref{eq2.1-1} with a given \( S \).}
          \added[id=HP]{\textbf{Evaluation datasets:} Once the designs are calibrated, new datasets are generated to evaluate PRN, PET, average sample size, and trial duration.  
          Event times in the control arm follow $\text{Exp}(\log 2 / \bar{\mu}_0)$, whereas event times in the experimental arm follow Equation~\eqref{eq2.1-1} with a \emph{fixed} true separation time $S$.}
\end{itemize}

\deleted[id=HP]{We considered the following scenarios:}
\added[id=HP]{Three delay scenarios were explored to cover a spectrum from no delay to substantial delay:}

\begin{enumerate}
  \item[1.] True $S = 0$ with hypotheses $H_0 : \bar{\mu}_0 = 3$ months and $H_1 : \bar{\mu}_1 = 6$ months.
  \item[2.] True $S = 1$ with hypotheses $H_0 : \bar{\mu}_0 = 3$ months and $H_1 : \bar{\mu}_1 = 5$ months.
  \item[3.] True $S = 2.25$ with hypotheses $H_0 : \bar{\mu}_0 = 3$ months and $H_1 : \bar{\mu}_1 = 3.75$ months.
\end{enumerate}

\deleted[id=HP]{For the BOP2 design, the optimal parameters and corresponding operating characteristics were obtained following the procedure described by \citet{Zhou:2020}. For DTE-BOP2, given \( \bar{\mu}_1 \) and \( S \), we derived \( \tilde{\mu}_0 \) and \( \tilde{\mu}_1 \) via Equation \eqref{relation}. We selected \( S_{\text{likely}} = 2.25 \) for all scenarios, with a truncated Gamma prior \( S \sim \text{Gamma}(1, 1) \cdot I(2.2, 2.5) \).}
\added[id=HP]{BOP2 parameters were tuned exactly as in \citet{Zhou:2020}.  
For DTE-BOP2, each pair $(\bar{\mu}_0,\bar{\mu}_1)$ and the true $S$ were first converted to post-delay medians $(\tilde{\mu}_0,\tilde{\mu}_1)$ using Equation~\eqref{relation}.  
A single design was calibrated with $S_{\text{likely}}=2.25$ months--the worst-case (longest) delay--and the truncated-Gamma prior $S\sim\mathrm{Gamma}(1,1)\,I(2.2,2.5)$. With no firm prior knowledge about $S$, we set $s_1=s_2=1$, yielding an (almost) uniform prior on $(L,U)$ for small $U-L$. This “one-design-fits-all” approach allows us to test robustness across shorter delays.}

\deleted[id=HP]{Unless expert information on the separation time \( S \) was available, we set \( s_1 = s_2 = 1 \) in \eqref{eq2.1_2} by default, resulting in a near-uniform distribution over \( (L,U) \) when \( U-L \) is small. For the priors on \( \mu_j \) (\( j = 0,1 \)) in Equation \eqref{eq2.1_4}, we used noninformative priors: \( \text{Inv-Gamma}(4, 9/\log2) \) for \( \mu_0 \) and \( \text{Inv-Gamma}(4, 18/\log2) \) for \( \mu_1 \).}
\added[id=HP]{For $\mu_0$ and $\mu_1$ we used weakly-informative priors $\mathrm{Inv\text{-}Gamma}(4,\,9/\log2)$ and $\mathrm{Inv\text{-}Gamma}(4,\,18/\log2)$, where the medians of the prior equal 3 and 6 months, respectively.}

The simulation results are presented in Table~\ref{Table_1}, leading to the following observations:
\begin{enumerate}
  \item[(a)] When there is no delayed treatment effect ($S=0$), BOP2 and DTE-BOP2 perform similarly.
  \item[(b)] As the true $S$ increases, BOP2 shows a marked loss of power, whereas DTE-BOP2 retains high power.
  \item[(c)] \added[id=HP]{DTE-BOP2 achieves comparable or smaller average sample size and shorter trial duration, while maintaining a lower PET under $H_1$, indicating earlier Go decisions for effective treatments.}
\end{enumerate}

\added[id=HP]{Overall, DTE-BOP2 is robust to the magnitude of the delay and consistently outperforms the original BOP2 design once a non-negligible delay is present.}

\deleted[id=HP]{In summary, the DTE-BOP2 design demonstrates superior performance compared to the original BOP2 design across various scenarios.}
\added[id=HP]{In summary, while both designs perform similarly when there is no delay in treatment effect, DTE-BOP2 substantially outperforms BOP2 in the presence of delay, offering notably higher power and more consistent performance across various delay severities. These results highlight the robustness and adaptability of the DTE-BOP2 design in handling time-to-event endpoints with potential delayed effects.
}

\begin{table}[H]
\centering
\resizebox{\textwidth}{!}{%
\caption{ \added[id=HP]{Comparison of key performance metrics (PRN, PET, average sample size, and trial duration) between DTE-BOP2 and BOP2 designs under various scenarios of delayed treatment effects. Each metric is averaged across 10,000 simulated trials. Sample size and trial duration are reported per arm and in months, respectively.
}  \deleted[id=HP]{Percentage of rejecting the null (PRN), percentage of early termination (PET), the actual sample size (each arm), and trial duration under the DTE-BOP2 and BOP2 design with a time-to-event endpoint in a two-arm randomized trial.}}
\begin{tabular}{|c|c|c|c|c|c|c|}
\hline
 & $(\bar{\mu}_0, \bar{\mu}_1)^\text{c}$ & Method & PRN (\%) & PET (\%) & Sample size & Trial duration$^\text{c}$ \\ \hline
\multirow{4}{*}{$S=0$} & \multirow{2}{*}{(3,3)$^\text{a}$} & BOP2 & 9.7 & 58.4 & 34.2 & 14.1 \\ \cline{4-7}
                       &                                   & DTE-BOP2 & 7.7 & 55.7 & 34.4 & 14.3 \\ \cline{2-7}
                       & \multirow{2}{*}{(3,6)$^\text{b}$} & BOP2 &91.0  &5.2  & 39.5 & 18.8 \\ \cline{4-7}
                       &                                   &  DTE-BOP2 &  90.0 & 4.4 & 39.6 & 18.9 \\ \hline
\multirow{4}{*}{$S=1$} & \multirow{2}{*}{(3,3)$^\text{a}$} & BOP2 &8.5  & 62.2 & 33.8 & 13.7 \\ \cline{4-7}
                       &                                   & DTE-BOP2  &  7.8 & 52.2 & 34.8 &14.6  \\ \cline{2-7}
                       & \multirow{2}{*}{(3,5)$^\text{b}$} & BOP2 &55.9  & 26.8 & 37.3 &16.9  \\ \cline{4-7}
                       &                                   &  DTE-BOP2 &  87.7& 5.6 & 39.4 & 18.9   \\ \hline
\multirow{4}{*}{$S=2.25$} & \multirow{2}{*}{(3,3)$^\text{a}$} & BOP2 & 9.8 & 70.0 & 33.0 & 13.0   \\ \cline{4-7}
                       &                                   & DTE-BOP2 &   8.0& 47.4 & 35.3 & 15.1  \\ \cline{2-7}
                       & \multirow{2}{*}{(3,3.75)$^\text{b}$} & BOP2 & 41.9 & 48.1 & 35.2 &15.0  \\ \cline{4-7}
                       &                                   &  DTE-BOP2  &  82.5& 7.7 & 39.2 &18.6  \\ \hline

\end{tabular}
\label{Table_1}
}

\begin{flushleft}
\footnotesize
\textsuperscript{a} Null hypothesis. \\
\textsuperscript{b} Alternative hypothesis. \\
\textsuperscript{c} In months.
\end{flushleft}
\end{table}

\subsection{Type-I Error and Power as Functions of the Separation Time}

\deleted[id=HP]{We further investigate the trends in type I error and power as functions of \( S \) using the DTE-BOP2 design.}
\added[id=HP]{To assess how reliably the calibrated design controls error across the entire prior range of the separation time, we investigated the type-I error and power curves of DTE-BOP2 as functions of the true \(S\).  Particular attention was paid to the effect of \emph{boundary control} -- that is, re-optimising the decision boundary \((\lambda,\gamma)\) so that the global type-I envelope rarely exceeds the nominal level -- versus a \emph{no-control} strategy that is tuned only for the average error.}

The following cases were examined:
\begin{enumerate}
  \item[1.] True \( S = 2.25 \) with hypotheses \( H_0: \bar{\mu}_0 = 3 \) months and \( H_1: \bar{\mu}_1 = 3.75 \) months.
  \item[2.] True \( S = 2.4 \) with hypotheses \( H_0: \bar{\mu}_0 = 3 \) months and \( H_1: \bar{\mu}_1 = 3.6 \) months.
\end{enumerate}

In both cases, we used \( S_{\text{likely}} = 2.25 \) and a truncated-Gamma prior \(S\sim\text{Gamma}(1,1)\,I(2.2,2.5)\).  Non-informative priors were adopted for \(\mu_j\) as described previously.  Two nominal type-I error thresholds were considered: \(\alpha = 0.10\) and \(\alpha = 0.15\).

\deleted[id=HP]{We plot the type I error and power across \( S \in [2.2, 2.5] \), as shown in Figures \ref{fig:1} and \ref{fig:1_1}. The results in Figures \ref{fig:1} and \ref{fig:1_1} illustrate the behavior of type I error and power under varying true separation times \(S\).}
\added[id=HP]{Figures~\ref{fig:1} and~\ref{fig:1_1} display the realised type-I error and power for a grid of true separation times \(S\in[2.2,2.5]\).  Each figure contains three sub-panels:  
(a)/(d) show the type-I error without boundary control;  
(b)/(e) show the type-I error with boundary control;  
(c)/(f) show the corresponding power curves (solid = no control, dashed = with control).}

\deleted[id=HP]{In Figure \ref{fig:1}, where the true \(S = 2.25\), the optimal parameters \(\lambda\) and \(\gamma\) obtained with and without boundary control coincide. Consequently, the type I error curves without and with boundary control are nearly identical, and both remain within the specified thresholds (\(\alpha = 0.1\) and \(\alpha = 0.15\)). Similarly, the power curves show negligible difference between the two strategies. In contrast, Figure \ref{fig:1_1}, with true \(S = 2.4\), demonstrates a clear distinction between the two approaches. Without boundary control, although the average type I error is controlled at the nominal level, the type I error exceeds the threshold at several separation times, particularly when \(\alpha = 0.15\). In contrast, with boundary control, the type I error remains strictly within the threshold across all separation times. While the boundary control strategy incurs a slight reduction in power, this loss is modest and acceptable given the improved error control.}
\added[id=HP]{When the true \(S\) equals the design value (Figure~\ref{fig:1}), the two calibration schemes coincide and both error curves lie comfortably below the nominal limits.  When the true \(S\) shifts to 2.4 (Figure~\ref{fig:1_1}), the no-control scheme exhibits local violations of the \(\alpha=0.15\) boundary around \(S\approx2.45\), whereas the boundary-control scheme keeps the entire curve below the red dashed line.  The safeguard costs no more than a 1 or 2\% loss of power, a modest trade-off for uniform error protection.}

\added[id=HP]{These findings highlight the value of boundary control when strict type-I regulation is required across a plausible range of separation times, without materially compromising power.}

\begin{figure}[H]
  \centering
  \resizebox{\textwidth}{!}{\includegraphics{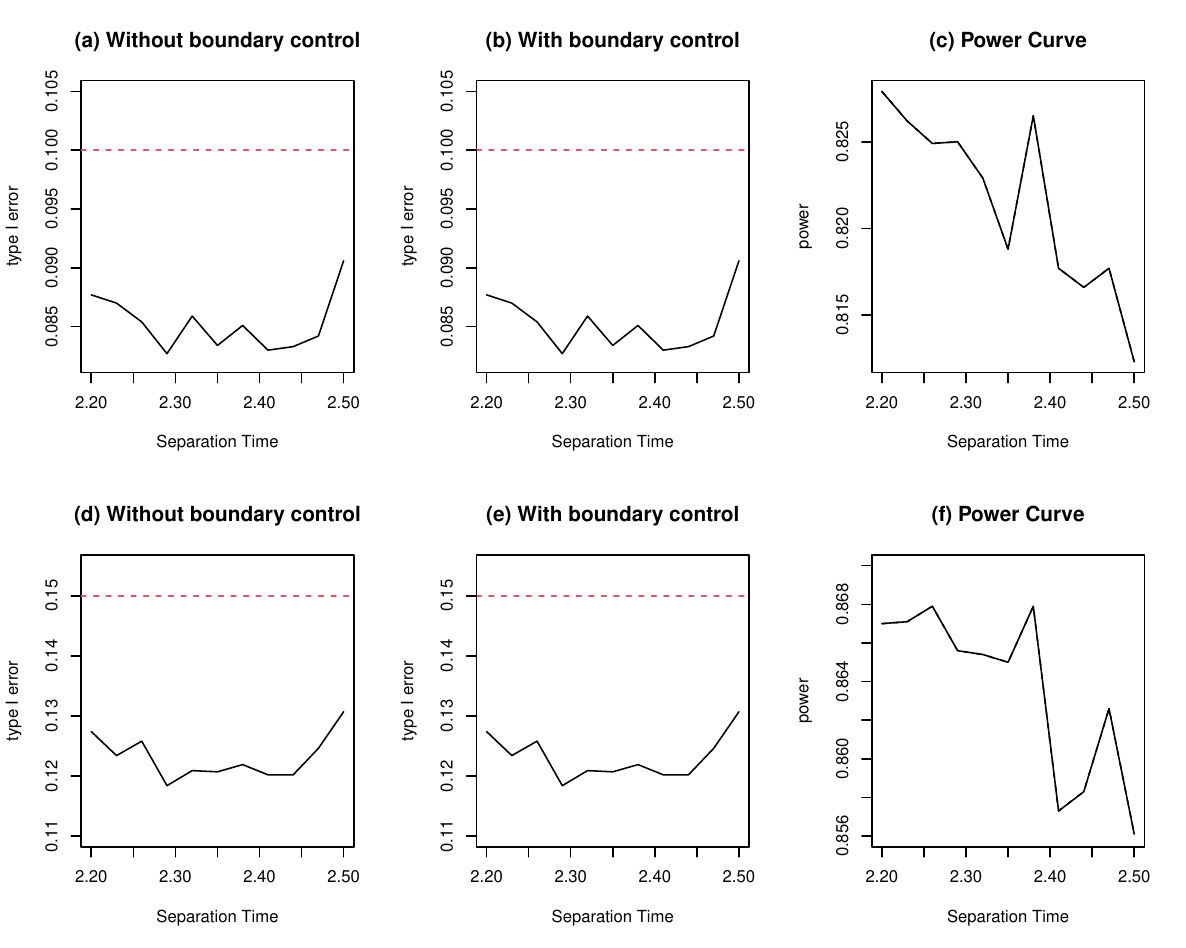}}
  \caption{True $S=2.25$.  Columns (a--c) correspond to a type-I error threshold of 0.10;  
           columns (d--f) correspond to 0.15.  
           Left: type-I error without boundary control;  
           middle: type-I error \emph{with} boundary control;  
           right: power curve (solid = no control, dashed = boundary control).}
  \label{fig:1}
\end{figure}

\begin{figure}[H]
  \centering
  \resizebox{\textwidth}{!}{\includegraphics{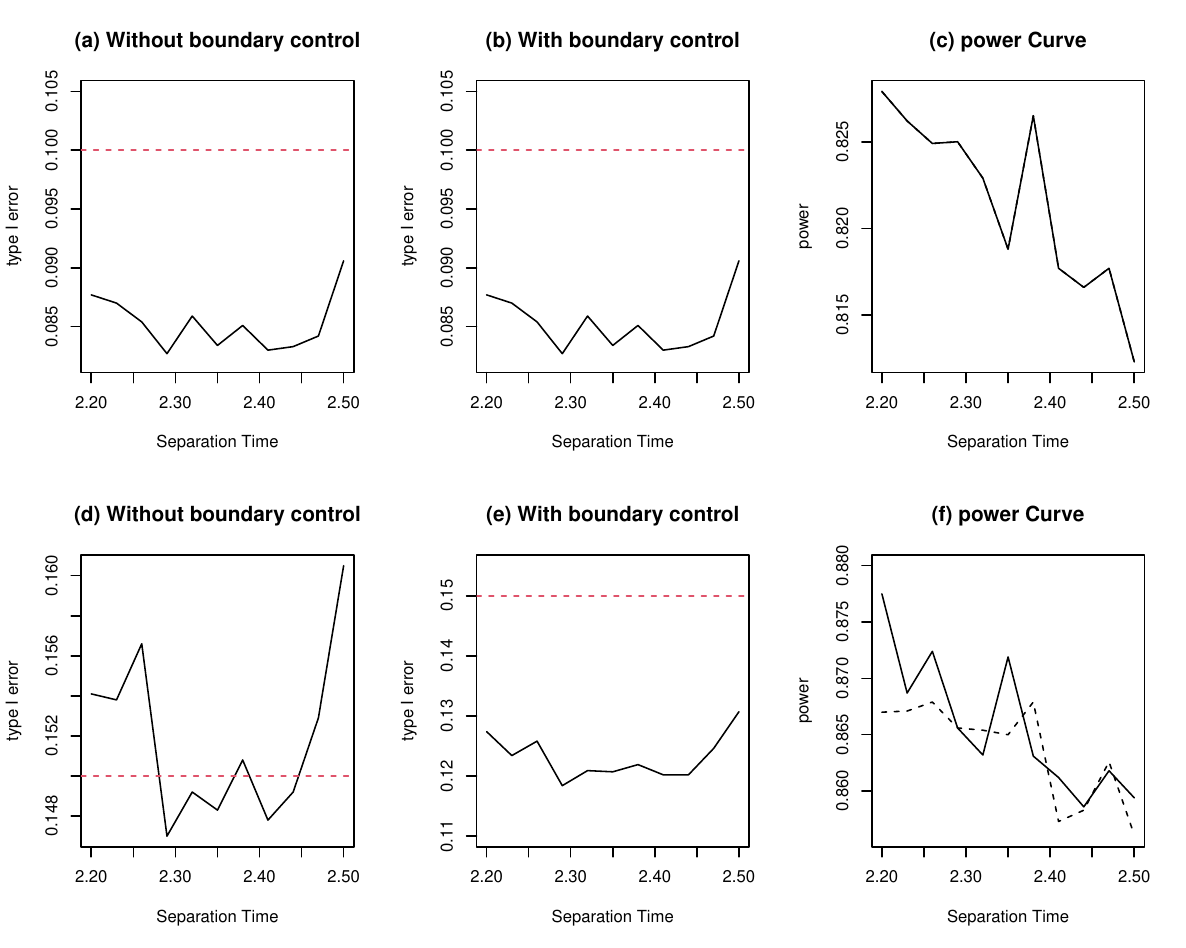}}
  \caption{True $S=2.4$.  Panel layout and line styles are identical to Figure~\ref{fig:1}.}
  \label{fig:1_1}
\end{figure}

\subsection{Sensitivity Analysis of the Choice of \(L\) and \(U\)}

\deleted[id=HP]{We performed a sensitivity analysis to examine the influence of different choices of the lower and upper bounds \( L \) and \( U \), respectively, under the default \( \text{Gamma}(1,1) \) prior specified in Equation \eqref{eq2.1_2}.}
\added[id=HP]{To assess the robustness of DTE\textendash BOP2 to the specification of the prior interval for the separation time~\(S\), we varied the lower and upper bounds \((L,U)\) while keeping the prior shape fixed at
\(S\sim\text{Gamma}(1,1)\,I(L,U)\)
as in Equation~\eqref{eq2.1_2}.  The \emph{control-- point} tuning algorithm was used to select \((\lambda,\gamma)\); it maximises the average power subject to a global two-sided type-I error constraint of~0.10.  Three interval widths,
\(U-L\in\{0.3,0.5,0.7\}\),
were examined, each at three different locations, yielding nine scenarios in total.}

The specific intervals were:
\begin{enumerate}
  \item[(1)] \(U-L=0.3\): \((L,U)=(2.4,2.7),\;(2.0,2.3),\;(1.8,2.1)\);
  \item[(2)] \(U-L=0.5\): \((2.4,2.9),\;(1.8,2.3),\;(1.6,2.1)\);
  \item[(3)] \(U-L=0.7\): \((2.4,3.1),\;(1.8,2.5),\;(1.6,2.3)\).
\end{enumerate}

\added[id=HP]{For each interval, the best-guess separation time was set to the midpoint,
\(S_{\text{likely}}=(L+U)/2\).  
The resulting optimal parameters and average operating characteristics are summarised in Table~\ref{tab:LU}.}

\vspace{1ex}
\begin{table}[H]
\centering
\caption{\added[id=HP]{Optimal tuning parameters and average operating characteristics under different prior intervals \((L,U)\).  Values in parentheses are (average type-I error, average power) averaged over \(S\sim\text{Gamma}(1,1)\,I(L,U)\).
Notably, the same pair \((\lambda,\gamma)=(0.95,1)\) is optimal for all scenarios, illustrating robustness to the choice of \(L\) and \(U\).  \(\tilde{\mu}_0=3\)~months.}}
\label{tab:LU}
\begin{tabular}{|c|c|c|c|}
\hline
$U-L$ & $(L,U)$ & $(\lambda,\gamma)$ & (AvgtypeI,\;AvgPower) \\ \hline
\multirow{3}{*}{0.3}
 & (2.4, 2.7) & (0.95, 1) & (0.087, 0.81) \\ \cline{2-4}
 & (2.0, 2.3) & (0.95, 1) & (0.088, 0.83) \\ \cline{2-4}
 & (1.8, 2.1) & (0.95, 1) & (0.084, 0.84) \\ \hline\hline
\multirow{3}{*}{0.5}
 & (2.4, 2.9) & (0.95, 1) & (0.089, 0.81) \\ \cline{2-4}
 & (1.8, 2.3) & (0.95, 1) & (0.081, 0.83) \\ \cline{2-4}
 & (1.6, 2.1) & (0.95, 1) & (0.079, 0.85) \\ \hline\hline
\multirow{3}{*}{0.7}
 & (2.4, 3.1) & (0.95, 1) & (0.084, 0.80) \\ \cline{2-4}
 & (1.8, 2.5) & (0.95, 1) & (0.085, 0.84) \\ \cline{2-4}
 & (1.6, 2.3) & (0.95, 1) & (0.086, 0.84) \\ \hline
\end{tabular}
\end{table}
\vspace{1ex}

\added[id=HP]{Two representative intervals were chosen for a deeper look at how type-I error and power change with the \emph{true} separation time \(S\):
a narrow prior \((U-L)=0.3\) and a wide prior \((U-L)=0.7\).
Figure~\ref{fig:LUcurves} displays the corresponding curves.}

\begin{figure}[H]
\centering
\subfigure[$U-L=0.3$]{\includegraphics[width=0.46\textwidth]{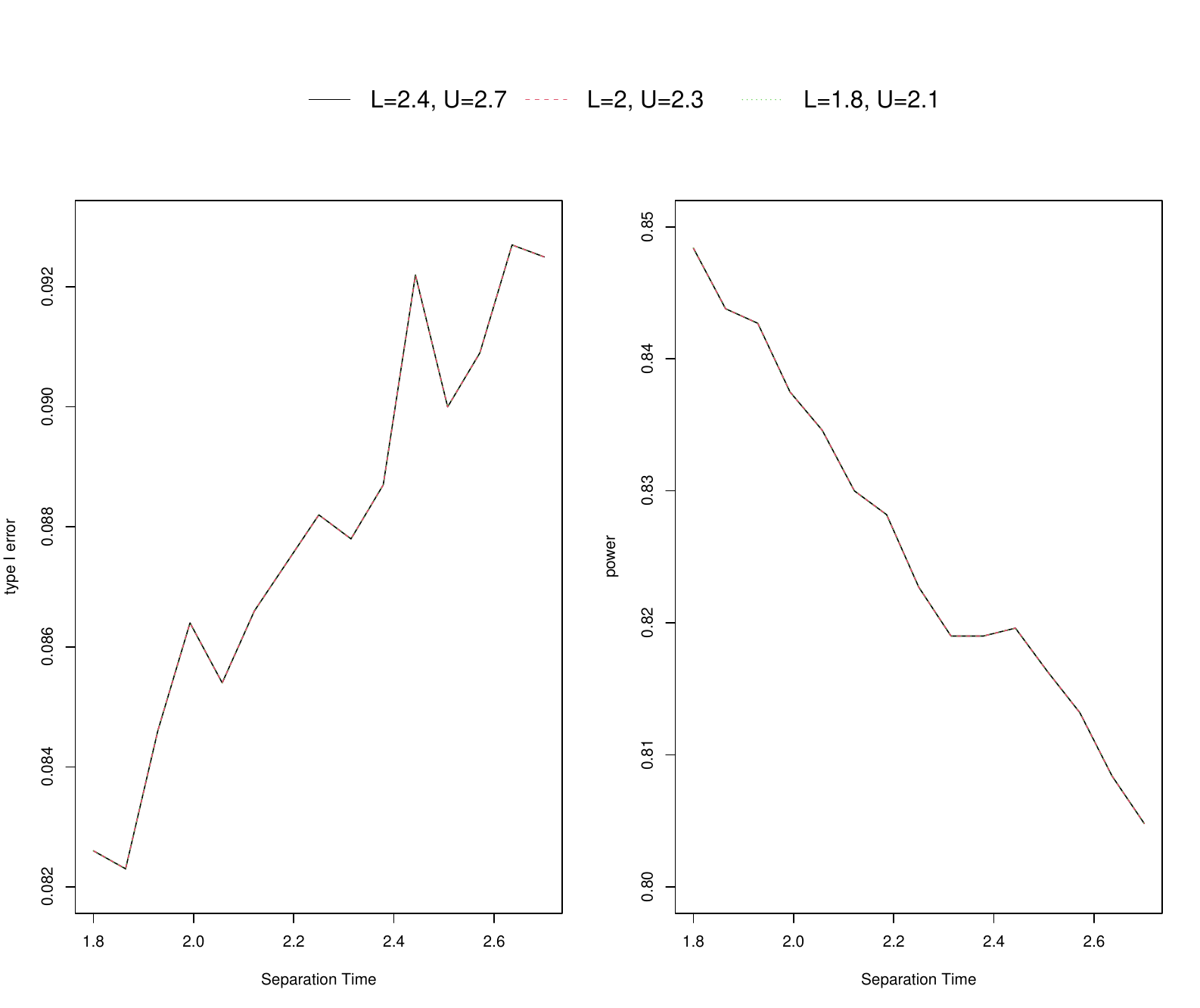}}
\hfill
\subfigure[$U-L=0.7$]{\includegraphics[width=0.46\textwidth]{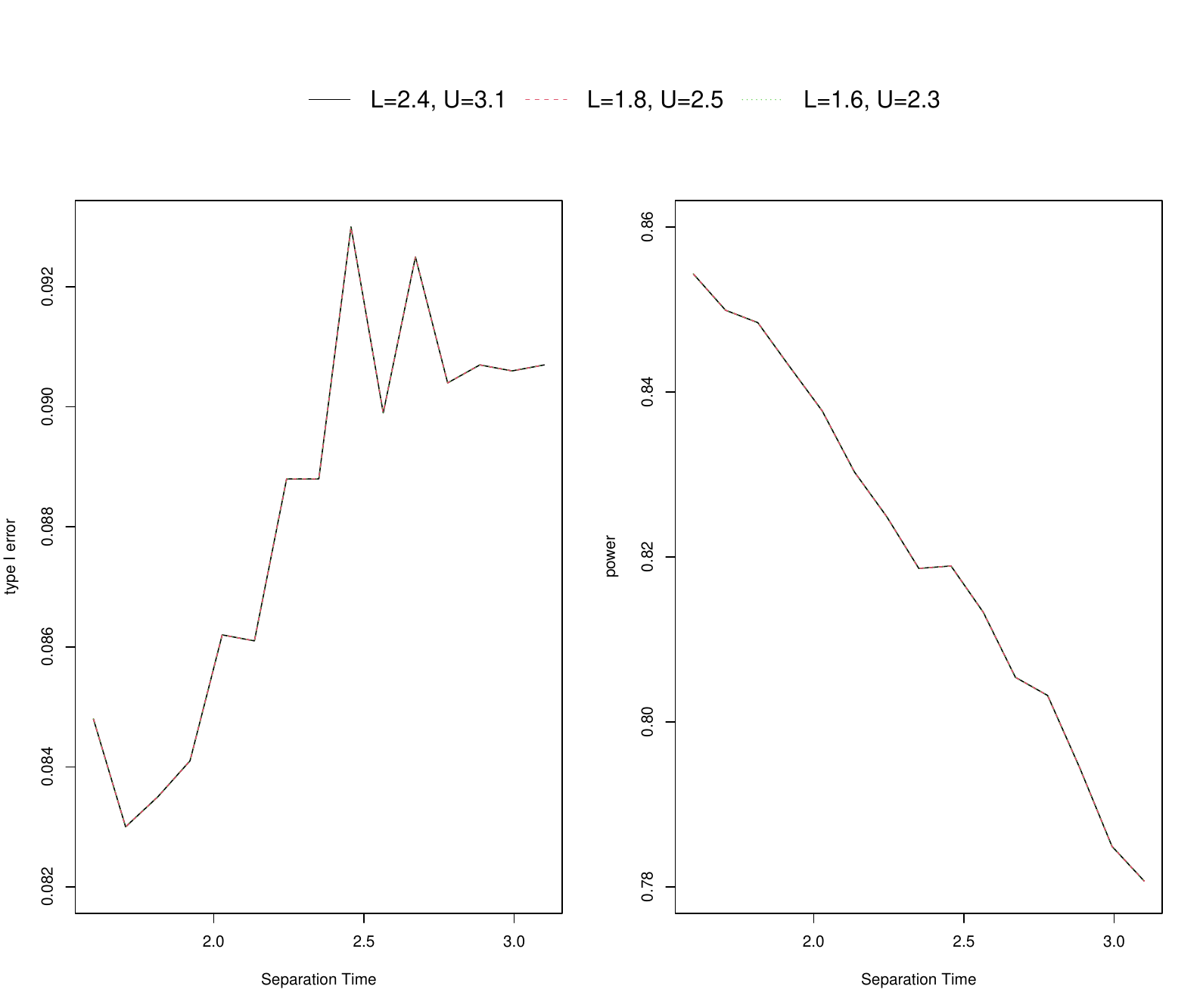}}
\caption{\added[id=HP]{Type-I error (left panel of each subfigure) and power (right panel) as functions of the true separation time \(S\).
The range of \(S\) plotted for each subfigure is bounded by the minimum \(L\) and maximum \(U\) in that group.}}
\label{fig:LUcurves}
\end{figure}

From Table~\ref{tab:LU} and Figure~\ref{fig:LUcurves} we draw three conclusions:
\begin{enumerate}
  \item Increasing the upper bound \(U\) shifts prior mass toward larger delays and therefore lowers the average power and power for given $S$.
  \item The optimal tuning parameters remain \((\lambda,\gamma)=(0.95,1)\) across all nine intervals, underscoring robustness to the choice of \((L,U)\).
  \item When the same tuning parameters apply, the type-I error and power curves are nearly identical across intervals; only the overall power level drifts downward as \(U\) grows.
\end{enumerate}

\added[id=HP]{A heuristic explanation for the power decline is provided in Appendix~A: as \(U\) increases, the quantile function in Equation~\eqref{formula} shrinks the effective post-delay information, reducing the signal-to-noise ratio.  Practically, this insight suggests choosing the prior interval as tight as prior knowledge allows, provided that the true \(S\) is believed to lie within \( [L,U] \).}

\section{Discussion}
\deleted[id=HP]{This work presents a flexible and efficient framework for designing two-arm multi-stage survival trials under delayed treatment effects (DTE). By incorporating prior knowledge about the likely separation timepoint $S$, our method enables substantial reductions in required sample size while preserving type I error control and achieving the desired statistical power. The use of the ratio $\tilde{\mu}_1/\tilde{\mu}_0$, based on posterior medians, instead of $\bar{\mu}_1/\bar{\mu}_0$, based on posterior means, provides a more stable and robust foundation for design optimization when DTE is anticipated.}

\added[id=HP]{We propose a comprehensive Bayesian decision framework for randomised phase-II survival trials in which the treatment effect is expected to emerge only after an unknown latency period.  
The design, termed DTE-BOP2, unifies four key ingredients:  
(i) an explicit piece-wise exponential likelihood that captures non-proportional hazards;  
(ii) a truncated-Gamma prior for the separation time~$S$, encoding expert or empirical knowledge of plausible delays;  
(iii) an adaptive rejection boundary that guarantees family-wise type-I error control over the entire interval $S\in[L,U]$; and  
(iv) decision rules formulated on the ratio of the posterior of \emph{medians} before and after separation time point, $\tilde\mu_1/\tilde\mu_0$, which is less sensitive to heavy-tailed posteriors than the classical overall median ratio $\bar\mu_1/\bar\mu_0$.  
Together, these components yield a time-driven design that simultaneously reduces sample size and preserves frequentist validity. The contributions of our work are as below:}

\deleted[id=HP]{A key methodological innovation of this work is the first introduction of a sample size calculation procedure for two-stage designs within a Bayesian framework. This contribution addresses a long-standing gap in trial planning under non-proportional hazards, offering a principled and data-informed alternative to traditional approaches. Moreover, our sensitivity analyses revealed that the calculated sample size for two-stage designs can offer practical guidance even when more than two stages are implemented. Although such extensions may not be strictly optimal, the resulting designs can still achieve power above the pre-specified threshold while maintaining type I error control. Another notable property of the proposed approach is its doubly robust performance across the prior interval $[L, U]$. Specifically, as long as both the true separation time $S$ and the best-guessed separation time $S_{\text{likely}}$ lie within $[L, U]$, the design reliably controls the type I error rate below the nominal level $\alpha$ and maintains statistical power above $1 - \beta$ throughout the interval. This property is especially advantageous in practice, where precise knowledge of $S$ is typically unavailable and misspecification is a concern.}
\begin{enumerate}
    \item[1.] \textbf{Sample-size determination under delayed effects}. To our knowledge, we provide the \emph{first closed-form algorithm} for calibrating interim and final sample sizes in a two-stage Bayesian design under non-proportional hazards.  Previous Bayesian proposals (e.g. \citet{Zhou:2020}) left sample size to ad-hoc rules or relied on proportional-hazard assumptions.  By optimising a weighted expected-sample-size criterion, DTE–BOP2 achieves near-minimal patient exposure under the null while maintaining high power under the alternative, and—importantly—retains these guarantees when the same boundary is embedded in designs with three or more looks.
    \item[2. ]\textbf{Doubly robust to misspecification.} When both the true delay~$S$ and the design-stage guess $S_{\text{likely}}$ reside in \([L,U]\), our simulation results show that type-I error never exceeds~$\alpha$, whereas the power remains above \(1-\beta\) throughout if we use the most conservative sample size strategy.  In contrast, methods that \emph{fix}~$S$ at a prespecified value--such as weighted log-rank tests or fixed-delay lag models--must conduct a post-hoc sensitivity analysis to assess vulnerability.  DTE–BOP2 internalises that uncertainty at the design stage.
    \item[3. ] \textbf{Software dissemination.} To facilitate implementation, we developed the \texttt{R} package \texttt{DTEBOP2}\footnote{\url{https://cran.r-project.org/package=DTEBOP2}}, which provides functions for sample size calculation, operating characteristic evaluation, and decision rule generation through concise syntax. The package also includes a vignette that reproduces all numerical studies presented in this paper, ensuring full reproducibility of results.
\end{enumerate}
\deleted[id=HP]{We also demonstrated that the operating characteristics of the proposed design are robust to the prior specification of $S$, including the width of the interval $[L, U]$. In comparative evaluations, the DTE-BOP2 design showed marked efficiency gains relative to the log-rank test, piecewise weighted log-rank test and the original BOP2 design. In scenarios with delayed treatment effects, DTE-BOP2 consistently achieved higher power, with the advantage becoming more pronounced as the actual separation time occurred later. These findings underscore the value of explicitly modeling delay-related uncertainty at the design stage. To promote practical implementation, we developed an open-source \texttt{R} package, \texttt{DTEBOP2} (available at \url{https://cran.r-project.org/web/packages/DTEBOP2/index.html}), which enables users to perform sample size calculation, simulate operating characteristics, and evaluate design performance under various assumptions. This tool enhances accessibility for practitioners and facilitates integration of our method into real-world applications.}

\added[id=HP]{Despite these contributions, our approach has several limitation. First, our piecewise exponential assumption can be relaxed to a Weibull or semi-parametric baseline is ongoing work.  
Second, while we provided two-stage sample-size formulas, a fully optimal multi-stage allocation remains an open problem; dynamic programming or reinforcement-learning tools may offer computationally tractable solutions.  
In addition, incorporating covariate-adjusted hazards could further improve efficiency in biomarker-driven trials.  These extensions will be explored in future research. Finally, uncertainty in patient accrual rates can materially affect operating characteristics.  
Explicitly modelling the accrual process--e.g., via a hierarchical prior on the enrollment intensity and propagating that uncertainty through the decision boundary--offers another promising route to further enhance the robustness of DTE–BOP2. Nevertheless, simulations based on a real-world example indicate that the operating characteristics remain acceptable even under plausible variability in accrual rates.}

\begin{comment}
    Despite these contributions, our approach has several limitations. First, we assume a single separation time point, which may be overly simplistic in clinical settings where treatment effects evolve in multiple phases. Future work could extend the framework to accommodate multiple or time-varying separation points. Second, the proposed method is based on a piecewise exponential distribution, which may not fully capture complex survival dynamics observed in real-world data. Incorporating more flexible baseline hazard models—such as spline-based or semi-parametric approaches—could enhance the generalizability of the design. Third, as demonstrated in our simulations, the accrual rate can affect the operating characteristics of the design. Incorporating prior information on the accrual rate, or explicitly modeling its uncertainty, may help construct more robust designs in practical applications where accrual patterns are unpredictable.
\end{comment}

\deleted[id=HP]{In conclusion, the DTE-BOP2 design provides a robust, computationally efficient, and conceptually transparent solution for designing clinical trials in the presence of delayed treatment effects. Its doubly robust properties, flexibility, and practical implementation make it a valuable tool for modern oncology and immunotherapy trials, where non-proportional hazards are both common and clinically meaningful.}

\added[id=HP]{In summary, DTE--BOP2 delivers a conceptually transparent yet computationally efficient framework for phase II trials with delayed treatment effects.  
It combines \emph{uniform type-I error control} with substantial sample-size savings, while its doubly robust decision rule confers protection against misspecification of the delay.  
Extensive simulations further demonstrate that the type I error remains stable across a wide range of separation times, while the power decreases monotonically as the separation time \( S \) increases. This trend provides practical guidance for controlling the power when \( S \in [L, U] \), where prior knowledge or historical data may inform plausible delay durations.  
Moreover, we find that the power is primarily governed by the \emph{ratio} of the median survival times before and after the separation point, rather than their absolute values, offering robustness and interpretability across a range of clinical scenarios.  
These features, together with freely available software, make the design particularly attractive for contemporary oncology and immuno-oncology studies in which non-proportional hazards are the rule rather than the exception.

}
\newpage
\appendix
\section{Trend Analysis of Type I error and Power with Separation Time}
To better understand the operating characteristics of our proposed design under delayed treatment effects, we conducted a comprehensive simulation study to evaluate the trends of type I error and power as a function of the separation time \( S \).  
Our results reveal that the type I error remains robustly controlled across a wide range of \( S \), consistently staying near the nominal level.  
In contrast, the power exhibits a monotonic decreasing trend as \( S \) increases, reflecting the diminishing early treatment information available for interim and final analyses. Furthermore, we observe that the power of the design is primarily driven by the \emph{ratio} between the median survival times before and after the separation time, rather than their absolute values.  
This insight highlights the relative nature of treatment effect accumulation in the presence of delayed onset and provides a more generalizable understanding of design behavior across scenarios with varying baseline hazards. To further support these findings, we provide an analytical explanation following the simulation results, which formalizes the observed relationships between separation time, median ratio, and power.
\subsection{Simulation}
Since the probability in Equation~\eqref{formula} is computed conditionally on the observed data, it represents a posterior probability that is itself a random quantity when the data are allowed to vary. Consequently, as we simulate datasets under different setting, this posterior probability induces a distribution, and it is natural to study its quantiles to assess its behavior under varying conditions. We investigate how the quantile of the probability
\begin{equation}
  P\left(K < \frac{b_0 + \text{TTOT}(I_0, I_1) + \sum_{i=1}^{n_r} z_{i0}}{b_0 + b_1 + \sum_{i=1}^{n_r} z_{i0} + \text{TTOT}(I_0, I_2)}\right)
  \label{app-1}
\end{equation}
varies with the separation time \( S \). Under the null hypothesis \( H_0 \), we assume \( \mu_0 = \mu_1 = \mu \). Conditional on observed data \( \mathcal{D}_{n_r} \) and a fixed \( S \), the posterior probability in \eqref{app-1} involves a Beta-distributed random variable \( K \sim \text{Beta}(a_0 + d_0 + d_{01}, a_1 + d_{11}) \). The threshold in the inequality depends on three components: \( \text{TTOT}(I_0, I_1) \), \( \text{TTOT}(I_0, I_2) = \text{TTOT}(I_0, I_1) + \text{TTOT}(I_1, I_2) \), and \( \sum_{i=1}^{n_r} z_{i0} \).

Following the setup in \citet{Han:2014}, these components follow:
\begin{align}
    \text{TTOT}(I_0, I_1) &\sim \text{Gamma}(d_{01}, \mu^{-1}), \nonumber \\
    \text{TTOT}(I_1, I_2) &\sim \text{Gamma}(d_{11}, \mu^{-1}), \nonumber \\
    \sum_{i=1}^{n_r} z_{i0} &\sim \text{Gamma}(d_0, \mu^{-1}).
    \label{app-2}
\end{align}
where $\text{Gamma}(a,b)$ has mean $a/b$. Because $\text{TTOT}(I_0, I_1)$ and $\text{TTOT}(I_1, I_2)$ are independent, it follows that \( \text{TTOT}(I_0, I_2)=\text{TTOT}(I_0, I_1)+\text{TTOT}(I_1, I_2) \sim \text{Gamma}(d_{01} + d_{11}, \mu^{-1}) \).

Under the alternative hypothesis \( H_1 \), we assume:
\begin{align}
    \text{TTOT}(I_0, I_1) &\sim \text{Gamma}(d_{01}, \mu^{-1}_0), \nonumber \\
    \text{TTOT}(I_1, I_2) &\sim \text{Gamma}(d_{11}, \mu^{-1}_1), \nonumber \\
    \sum_{i=1}^{n_r} z_{i0} &\sim \text{Gamma}(d_0, \mu^{-1}_0).
    \label{app-3}
\end{align}
To simplify the simulation setting, we ignore administrative censoring and interim analyses, assuming all deaths are observed, i.e., \( n_r = N \), \( d_0 = n_r \), and \( d_{01} + d_{11} = n_r \). We allow \( d_{01} \in [0, n_r] \) to vary continuously to represent the effect of different separation times. This allows us to analyze how the probability in \eqref{app-1} changes as the number of events before separation increases. Fractional values of \( d_{01} \) and \( d_{11} \) are used for smooth trend estimation.

We discretize \( d_{01} \) into 100 equally spaced values within \([0, n_r]\) to mimic the variation of \( S \). For a fixed quantile level \( B \), we compute the \( B \)th quantile of
\begin{equation}
    P\left(K < \frac{b_0 + \text{TTOT}(I_0, I_1) + \sum_{i=1}^{n_r} z_{i0}}{b_0 + b_1 + \sum_{i=1}^{n_r} z_{i0} + \text{TTOT}(I_0, I_2)}\right)
    \label{quantile}
\end{equation}
at each \( d_{01} \), which serves as a proxy for \( S \). This quantile can be interpreted as the type I error (under \( H_0 \)) or power (under \( H_1 \)). The procedure is as follows:

\begin{enumerate}
    \item[Step 1] Under \( H_0 \), for each value of \( d_{01} \), simulate 10,000 samples of \( \text{TTOT}(I_0, I_1) \), \( \text{TTOT}(I_1, I_2) \), and \( \sum z_{i0} \) using \eqref{app-2}. Compute the probability in \eqref{app-1} for each sample.
    \item[Step 2] Compute the \( B \)th quantile of the 10,000 simulated probabilities.
    \item[Step 3] Repeat Steps 1-2 across all values of \( d_{01} \in [0, n_r] \).
    \item[Step 4] Repeat Steps 1-3 under \( H_1 \), generating samples from \eqref{app-3}.
\end{enumerate}

We use parameters derived from the real-data example: \( a_0 = 4 \), \( a_1 = 4 \), \( b_0 = 12.12 \), \( b_1 = 24.24 \), \( \mu_0 = 2.8/\log(2) \), and \( \mu_1 = 6.57/\log(2) \) (under $H_1$). We consider sample sizes \( n_r \in \{40, 50, 60, 70\} \) and quantile levels \( B \in \{0.025, 0.05, 0.1\} \). 

To evaluate the sensitivity of the \( B \)th quantile in Equation~\eqref{quantile} to the separation time \( S \), we treat the parameter \( d_{01} \in [0, n_r] \) as a proxy for \( S \), with 100 equally spaced values used to approximate a continuous distribution. We then assess the range of quantile values over two specific subintervals of \( d_{01} \): the interquartile range \( [Q_{25}, Q_{75}] \) the wider range \( [Q_{10}, Q_{90}] \). For each interval, we compute the minimum and maximum of the \( B \)th quantile, as well as the difference between them. Additionally, we calculate the Spearman correlation coefficient between the quantile values and the corresponding \( d_{01} \) values within each interval to assess monotonicity. These metrics provide a comprehensive picture of how sensitive the probability in \eqref{quantile} is to the timing of the treatment effect onset, as represented by \( S \). Results are reported in Table~\ref{app-table}.

From these results, we observe the following:
\begin{itemize}
    \item Under \( H_0 \), the $B$th quantile demonstrates limited variability across different values of \( S \), with the difference between the maximum and minimum values consistently bounded within 0.05. This indicates that the average type I error serves as a robust summary measure, accurately reflecting its behavior over varying separation times.
    \item Under \( H_1 \), the Spearman correlation is consistently close to $-1$, revealing a nearly perfectly monotonic decline in power as the proportion of events occurring before the separation time increases.
\end{itemize}

To further illustrate the trends observed in our simulations, Figure~\ref{fig:app_1} displays the $B = 0.05$ quantile for \( n_r = 50 \) and \( n_r = 60 \). These plots confirm our earlier findings: the type I error remains stable across most values of $S$, while power exhibits a clear and consistent downward trend. We also explored different combinations of $\mu_0$ and $\mu_1$, and found that as long as the ratio $\mu_1 / \mu_0$ remains constant, the overall patterns of type I error and power are essentially unchanged.

To provide a deeper understanding of this phenomenon, we now analyze the analytical form underlying the right-hand side of Equation~\eqref{formula}. Under the default prior specification, where $a_0 = a_1 = 4$, $b_0 = 3\mu_0$, and $b_1 = 6\mu_0$, the expression simplifies to:
\begin{align}
    \frac{b_0 + \text{TTOT}(I_0, I_1) + \sum_{i=1}^{n_r} z_{i0}}{b_0 + b_1 + \sum_{i=1}^{n_r} z_{i0} + \text{TTOT}(I_0, I_2)}
    &= \frac{3 + \mu_0^{-1} \text{TTOT}(I_0, I_1) + \mu_0^{-1} \sum_{i=1}^{n_r} z_{i0}}{9 + \mu_0^{-1} \left( \sum_{i=1}^{n_r} z_{i0} + \text{TTOT}(I_0, I_1) + \text{TTOT}(I_1, I_2) \right)} \nonumber \\
    &= \frac{3 + X + Y}{9 + X + Y + \frac{\tilde{\mu}_1}{\tilde{\mu}_0} Z}, \label{ratio}
\end{align}
where $X \sim \text{Gamma}(d_{01}, 1)$, $Y \sim \text{Gamma}(d_0, 1)$, and $Z \sim \text{Gamma}(d_{11}, 1)$.

This formulation yields two key insights:
\begin{itemize}
    \item The right-hand side of Equation~\eqref{formula} depends solely on the ratio $\mu_1 / \mu_0 = \tilde{\mu}_1 / \tilde{\mu}_0$ and is invariant to the absolute magnitudes of $\tilde{\mu}_0$ and $\tilde{\mu}_1$.
    \item Consequently, for a fixed sample size and a given value of $\tilde{\mu}_1 / \tilde{\mu}_0$, the qualitative behavior of both type I error and power remains consistent.
\end{itemize}

\begin{figure}
    \centering
    \resizebox{\textwidth}{!}{\includegraphics{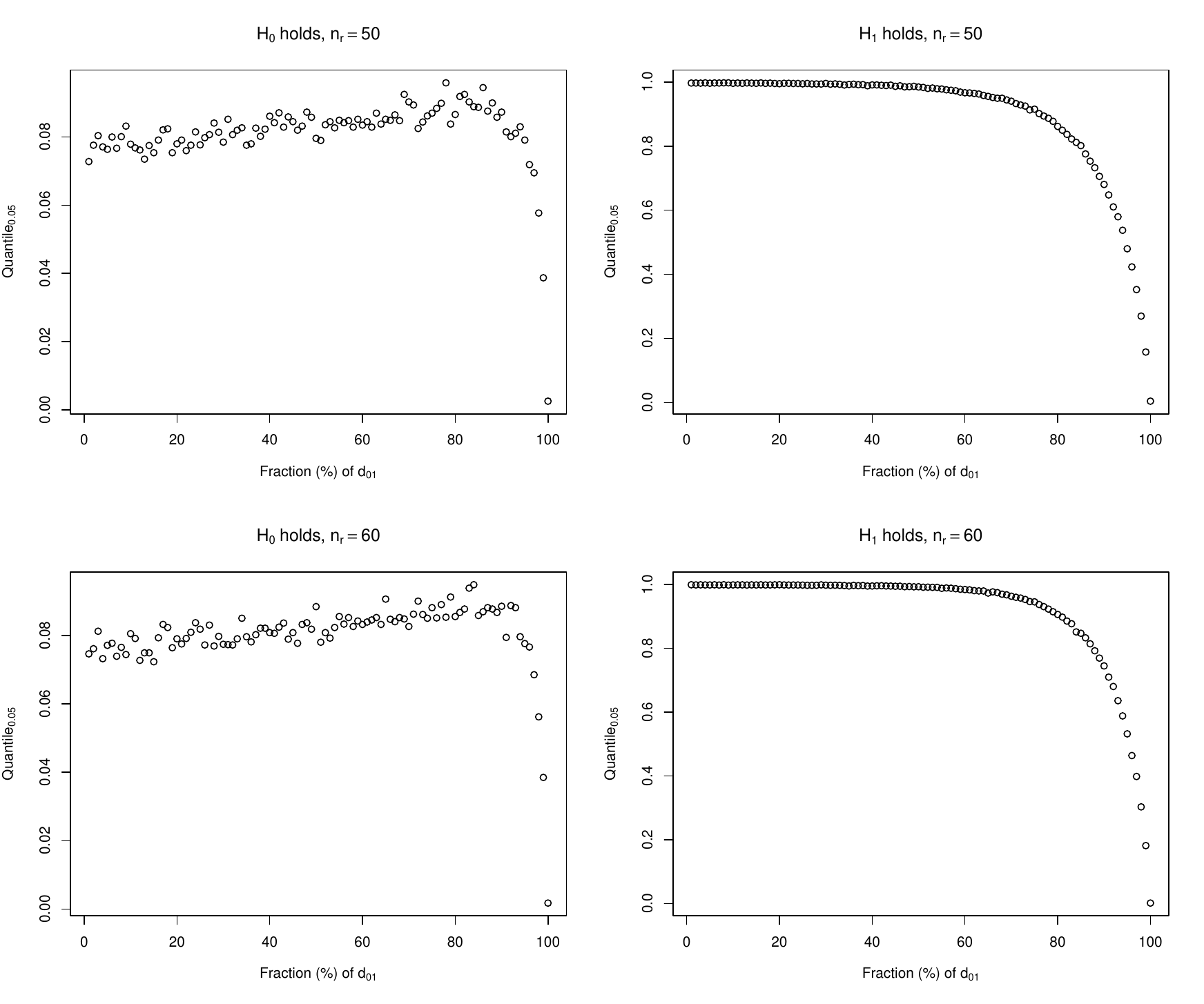}}
   \caption{The 0.05 quantile of \eqref{formula} under $H_0$ and $H_1$ for sample sizes of 50 and 60.}
    \label{fig:app_1}
\end{figure}

\begin{table}[ht]
\centering
\begin{tabular}{|c|c|c|c|c|c|c|c|c|c|}
\hline
\multicolumn{10}{|c|}{\textbf{Under $H_0$}} \\
\hline
\multirow{2}{*}{$n_r$} & \multirow{2}{*}{$B$} & \multicolumn{4}{c|}{$[Q_{25},Q_{75}]$} & \multicolumn{4}{c|}{$[Q_{10},Q_{90}]$} \\
\cline{3-10}
                       &                      & min & max & max-min & Spearman $\rho$& min & max & max-min & Spearman $\rho$ \\
\hline
\multirow{3}{*}{40} & 0.025 &  0.039 & 0.049  &  0.01 & 0.17  &  0.049 & 0.038  & 0.012  & 0.067  \\
                    & 0.05  &    0.08    &   0.09     &  0.01      &   0.62    &     0.08   &   0.09     &   0.01     &   0.79     \\
                    & 0.1   &     0.16   &   0.19     &  0.03      &    0.81   &  0.15      & 0.20       &    0.05    &   0.92     \\
\hline
\multirow{3}{*}{50} & 0.025 &     0.037   &   0.047     &   0.01     &   0.44    &    0.037    &   0.047     &   0.01     &  0.42    \\
                    & 0.05  &     0.079   &   0.091     &   0.012     &     0.49  &  0.076      &  0.094      &  0.018      &    0.72    \\
                    & 0.1   &     0.15   &  0.18      &  0.03      &    0.9   &   0.14     &  0.19      &   0.05     &     0.92   \\
\hline
\multirow{3}{*}{60} & 0.025 &    0.037    &  0.047      &    0.01    &   0.52    &    0.036    &    0.047    &   0.012 & 0.59   \\
                    & 0.05  &      0.075  &   0.091     &  0.016      &  0.65     &    0.073    &   0.093     &  0.02      &   0.77     \\
                    & 0.1   &     0.15   &   0.17     & 0.02       &    0.84   &     0.14   &   0.19     &  0.05      &    0.94    \\
\hline
\multirow{3}{*}{70} & 0.025 &     0.037   &  0.049      &   0.02     &    0.57   &    0.037    & 0.049   &  0.02  & 0.55       \\
                    & 0.05  &    0.074    &   0.06     &  0.014      &   0.60    &  0.071      & 0.092       &    0.021    &      0.81  \\
                    & 0.1   &    0.14    &   0.17     &  0.03      &    0.78   &  0.14      &   0.18     &    0.04    &     0.91   \\
\hline
\hline
\multicolumn{10}{|c|}{\textbf{Under $H_1$}} \\
\hline
\multirow{3}{*}{40} & 0.025 &    0.78    &  0.97     &   0.19     &    -0.99   &     0.49   &  0.98      &  0.48 &    -0.99    \\
                    & 0.05  &    0.86   &   0.98    &   0.12    &  -0.99    &    0.61   &  0.99      &  0.38      &  -0.99     \\
                    & 0.1   &        0.92 &    0.99    & 0.07       &    -0.99   &   0.75     &  0.99      &   0.24     &  -0.99          \\
\hline
\multirow{3}{*}{50} & 0.025 &     0.86  &   0.99     &    0.13    &   -0.99    &     0.56   &  0.99      &    0.44    &  -0.99        \\
                    & 0.05  &    0.91    &   0.99     &  0.08      &   -0.99    &   0.67     &   0.99     &   0.32     &   -0.99     \\
                    & 0.1   &     0.95   &   0.99     &   0.04     &  -0.98     &    0.80    &  0.99      &  0.19      &   -0.99     \\
\hline
\multirow{3}{*}{60} & 0.025 &  0.90      &    0.99    &    0.09    &   -0.99    &    0.63    &   0.99   &   0.36    &    -0.99    \\
                    & 0.05  &  0.94      &  0.99      &    0.05    &  -0.99     &   0.73     &  0.99      &  0.26      &   -0.99     \\
                    & 0.1   &   0.97     &   0.99     & 0.02       &    -0.98   &    0.84    &    1    &   0.16     & -0.98       \\
\hline
\multirow{3}{*}{70} & 0.025 &   0.94     &   0.99     &   0.05     &   -0.99    &   0.69     &  0.99      &    0.3    &  -0.99   \\
                    & 0.05  &    0.96    &  0.99      &   0.03     &   -0.99    &     0.78   &  1      &    0.22    &     -0.99   \\
                    & 0.1   &   0.98    &   1     &   0.02   &   -0.98   &  0.88    &   1   &    0.12    &   -0.98     \\
\hline
\end{tabular}
\caption{Quantile results of \eqref{quantile} under $H_0$ and $H_1$ across $[Q_{25}, Q_{75}]$ and $[Q_{10}, Q_{90}]$ intervals for varying $n_r$ and $B$ values.}
\label{app-table}
\end{table}

%\begin{figure}
%    \centering
%    \includegraphics[scale=0.8]{Ushape.pdf}
%    \caption{The trend of $P(\tilde{\mu}_1< \tilde{\mu}_0 \mid \mathcal{D}_{n_r}, S)$, $x$ axis is the proportion$\times$100 of death number. Red lines are the natural cubic splines with degree freedom 4}
%    \label{fig:appendix1}
%\end{figure}
 \newpage
 	\renewcommand{\baselinestretch}{1.0} \tiny\normalsize
	\setlength{\parskip}{1.0mm}
	\bibliography{BGS}
	\bibliographystyle{agsm}
\end{document}